\documentclass[12pt]{article}
\usepackage[a4paper]{geometry}
\usepackage{babel}
\pdfoutput=1

\usepackage{fancyhdr}
\usepackage{amsmath}
\usepackage{amsfonts}
\usepackage{amstext}
\usepackage{amssymb}
\usepackage{amsthm}
\usepackage{amscd}

\usepackage[pagebackref,draft=false]{hyperref}
\hypersetup{colorlinks,
linkcolor=myrefcolor,
citecolor=mycitecolor,
urlcolor=myurlcolor}

\usepackage[capitalize]{cleveref}
\usepackage{caption}
\usepackage{etaremune}

\usepackage{xcolor}
\definecolor{myurlcolor}{rgb}{0,0,0.4}
\definecolor{mycitecolor}{rgb}{0,0.5,0}
\definecolor{myrefcolor}{rgb}{0.5,0,0}
\usepackage{graphicx}
\usepackage{tikz}
\usepackage{tikz-cd}
\usepackage{mathrsfs}

\usepackage{etoolbox}
\usepackage{makeidx}
\usepackage{sectsty}
\usepackage{dsfont}
\usepackage{enumitem} 
\usepackage[]{latexsym}
\usepackage{braket}
\usepackage{caption}
\usepackage[utf8]{inputenx}
\usepackage[T1]{fontenc}
\usepackage{lmodern}
\usepackage{textcomp}
\usepackage{microtype}
\usepackage{totcount}
\usepackage{blindtext}
\newtheorem{remark}{Remark}

\newtheorem{proposition}{Proposition}

\newtheorem*{proof*}{Proof}

\newcommand{\be}{\begin{equation}}
\newcommand{\ee}{\end{equation}}
\newcommand{\bea}{\begin{eqnarray}}
\newcommand{\eea}{\end{eqnarray}}
\newcommand{\vsp}{\vspace{0.4cm}}
\newcommand{\grit}[1]{{\bfseries {\itshape {#1}}}}

\newcommand{\blue}[1]{\color{blue}{{#1}}}

\newcommand{\cfr}[1]{({\itshape cfr.} {#1})}

\newcommand{\ra}{\rightarrow}

\newcommand{\Lra}{\Longrightarrow}

\newcommand{\hh}{\mathcal{H}}
\newcommand{\bh}{\mathcal{B}(\mathcal{H})}
\newcommand{\bhsa}{\mathcal{B}_{sa}(\mathcal{H})}
\newcommand{\Glh}{\mathcal{GL}(\mathcal{H})}

\newcommand{\Uh}{\mathcal{U}(\mathcal{H})}
\newcommand{\uh}{\mathfrak{u}(\mathcal{H})}

\newcommand{\stsph}{\mathscr{S}(\mathcal{H})}
\newcommand{\Tr}{\textit{Tr}}
\newcommand{\posh}{\mathcal{P}(\mathcal{H})}

\newcommand{\stsp}{\mathscr{S}}

\newcommand{\gapp}{\mathscr{G}}

\newcommand{\pos}{\mathscr{P}}

\newcommand{\gr}{\mathrm{g}}

\title{Group actions and Monotone Quantum Metric Tensors}

\author{F. M. Ciaglia$^{1,3}$  \href{https://orcid.org/0000-0002-8987-1181}{\includegraphics[scale=0.7]{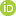}}, F. Di Nocera$^{2,4}$  \href{https://orcid.org/0000-0002-1415-2422}{\includegraphics[scale=0.7]{ORCID.png}}\\
\footnotesize{$^{1}$\textit{Universidad Carlos III de Madrid,  28903  Madrid, Spain}} \\
\footnotesize{$^{2}$\textit{Max Planck Institute for Mathematics in the Sciences,  04103  Leipzig, Germany}} \\
\footnotesize{$^{3}$\textit{ e-mail: \texttt{f.ciaglia[at]math.uc3m.es}} $^{4}$\textit{ e-mail: \texttt{dinocer[at]mis.mpg.de}}} \\
}

\begin{document}

\maketitle

\begin{abstract}
The interplay between actions of Lie groups   and monotone quantum metric tensors on the space of faithful quantum states of a finite-level system  observed recently in \href{https://doi.org/10.1140/epjp/s13360-020-00537-y}{DOI: 10.1140/epjp/s13360-020-00537-y} and \href{https://doi.org/10.1007/978-3-030-80209-7_17}{DOI: 10.1007/978-3-030-80209-7\textunderscore 17} is here further developed.
%%%%
\end{abstract}

\tableofcontents
\thispagestyle{fancy}

\section{Introduction} \label{sec:introduction}
%%%%%%%%%%%%%%%%%%%%%%%%%%%%%%
%%%%%%%%%%%%%%%%%%%%%%%%%%%%%%

In the context of information geometry for finite-dimensional quantum systems, it is well-known that the canonical action $\rho\mapsto U\rho U^{\dagger}$ of the unitary group $\Uh$ on the manifold $\stsp(\hh)$ of faithful quantum states  provides symmetry transformations for every monotone quantum metric tensor on $\stsp(\hh)$ pertaining to Petz's classification \cite{Petz-1996}.
%%%%
Therefore, the fundamental vector fields generating the canonical action of $\Uh$ are Killing vector fields for every quantum monotone metric tensor.
%%%%%%

It is also known that the canonical action of $\Uh$ on $\stsp(\hh)$ can be seen as the restriction to $\Uh$ of a nonlinear action of the general linear group $\Glh$ given by
\begin{equation}  \label{eqn:action of general linear group on states BH}
	\rho\mapsto \beta(\gr,\rho)= \frac{\gr \, \rho \, \gr^{\dagger}}{\mathrm{Tr}(\gr \, \rho \, \gr^{\dagger})}.
\end{equation}
%%%
The action $\beta$ is transitive on $\stsp$ and turns it into a homogeneous manifold \cite{C-DC-I-L-M-2017,C-J-S-2020,G-K-M-2005,G-K-M-2006}.
%%%
Therefore, the fundamental vector fields of the canonical action of $\Uh$ form a Lie-subalgebra of the algebra of fundamental vector fields of the action of $\Glh$.
%%%%%%% 

In \cite{Ciaglia-2020}, it is showed that, in order to describe the fundamental vector fields of $\beta$, it is sufficient to consider the fundamental vector fields of the canonical action of $\Uh$ on $\stsp(\hh)$ together with  the gradient vector fields associated with the expectation-value functions $l_{a}(\rho)=Tr(a\rho)$ - where $\mathbf{a}$ is any self-adjoint element  in the space $\bh$ of bounded linear operators on $\hh$ - by means of the so-called Bures-Helstrom metric tensor \cite{B-Z-2006, Bures-1969,   Dittmann-1995, Helstrom-1967, Helstrom-1968, Uhlmann-1992}.
%%%%%
This instance provides an unexpected link between the the unitary group $\Uh$, the $\Glh$-homogeneous manifold structure of $\stsp(\hh)$, the Bures-Helstrom metric tensor, and the expectation value functions.
%%%%

However, this is not the only example in which a monotone metric tensor “interacts” with the general linear group $\Glh$.
%%%%
Indeed, again in \cite{Ciaglia-2020}, it is also showed that fundamental vector fields of the canonical action of $\Uh$ together with the gradient vector fields associated with the expectation value functions by means of the Wigner-Yanase metric tensor \cite{G-I-2001,G-I-2003,Hasegawa-1993,Hasegawa-1995,Hasegawa-2003,H-P-1997,Jencova-2003-2} close on a representation of the Lie algebra of $\Glh$ that integrates to a group action given by
\begin{equation} \label{eqn:action of general linear group on states WY}
\rho \mapsto  \beta^{WY} ( \gr, \rho) = \frac{ (\gr \sqrt{\rho} \gr^{\dagger})^2}{\Tr\left( \left( \gr \sqrt{\rho} \gr^{\dagger} \right)^2 \right)} \in \stsph.
\end{equation}
%%%% 
Of course, the action $\beta^{WY} $   is different from the action $\beta$, but it is still a transitive action so that $\stsp(\hh)$ is a homogeneous manifold also with respect to this action, and the underlying smooth structure coincides with the one related with $\beta$.
%%%%%
Moreover, a direct inspection shows that $\beta^{WY} $   can be thought of as a kind of deformation of $\beta$  by means of the square-root map and its inverse on positive operators.
%%%%
This instance is better described and elaborated upon in the rest of the paper.
%%%%

Finally, again in \cite{Ciaglia-2020}, it is proved that there is another Lie group “extending” the unitary group $\Uh$ and for which a construction similar to the one discussed above is possible.
%%%%
This Lie group is the cotangent bundle $T^{*}\Uh$ of $\Uh$ endowed with its canonical Lie group structure \cite{A-G-M-M-1994,A-G-M-M-1998}.
%%%%%
In this case, the gradient vector fields of the expectation value functions are built using the Bogoliubov-Kubo-Mori metric tensor \cite{Naudts-2018,N-V-W-1975,Naudts-2021,Petz-1994,P-T-1993}, and the action is given by
\begin{equation}  \label{eqn:action of general linear group on states BKM}
	\rho\mapsto \gamma((\mathbf{U},\mathbf{a}),\rho)= \frac{ \mathrm{e}^{U \ln(\rho) U^\dagger + \mathbf{a}}}{\Tr( \mathrm{e}^{U\ln(\rho) U^\dagger + \mathbf{a}})},
\end{equation}
%%%
where $\mathbf{a}$ is a self-adjoint element which is identified with a cotangent vector at $U$.
%%%%%
Once again, we obtain a transitive action on $\stsp(\hh)$ associated with a homogenous manifold structure whose underlying smooth structure coincides with the two others smooth structures previously mentioned.
%%%%%%

It is important to note that, when we restrict to the unitary group $\Uh$, all the group actions we considered reduce to the canonical action
%%%%%%%%%%%%%%%%%%%%%%%%%%%%%%
\begin{equation} \label{eqn:definition_alpha_stsp}
( U, \rho) \mapsto \alpha(U,\rho) = U \rho U^* 
\end{equation}
%%%%%%%%%%%%%%%%%%%%%%%%%%%%%%
of the unitary group whose importance in quantum theories is almost impossible to overestimate.
%%%%%%%

Once we have these three “isolated” instances, it is only natural to wonder if they are truly isolated cases, or if there are other monotone metric tensors for which a similar construction is possible.
%%%%
In \cite{C-DN-2021} this problem is completely solved in the two-level case in which a direct, coordinate-based solution is possible.
%%%%
The result is that the only two groups for which the aforementioned construction works are precisely the general linear group $\Glh$ and the cotangent group $T^{*}\Uh$.
%%%%
Moreover, in the case of $T^{*}\Uh$, the only compatible action is the action $\gamma$ already described in \cite{Ciaglia-2020}, while, for $\Glh$, there is an entire family of compatible smooth actions parametrized by a real number $\kappa\in (0,1]$ and given by
\begin{equation} \label{eqn:family of actions of special linear group on states}
  \rho \mapsto \beta^{\kappa}(\gr, \rho) = \frac{ (\gr \rho^{\sqrt{\kappa}} \gr^\dagger)^{\frac{1}{\sqrt{\kappa}}}} {\Tr\left( \left( \gr \rho^{\sqrt{\kappa}} \gr^\dagger \right)^{ \frac{1}{\sqrt{\kappa}}} \right)}.
\end{equation}
%%%
All these actions are connected with a different   quantum  metric tensor.
%%%%
For instance, when $\kappa=1$ the Bures-Helstrom metric tensor is recovered, while the Wigner-Yanase metric is recovered when $\kappa=1/4$.
%%%%
Al the other cases correspond to Riemannian metric tensors on $\stsp(\hh)$ which are invariant under the standard action of $\Uh$.
%%%%%

In this work we further investigate the problem   by showing that all the group actions and metric tensors found in \cite{C-DN-2021} for a two-level system actually appear also for a quantum system with an arbitrary, albeit finite, number of levels.
%%%%
Moreover, we characterize all the values of $\kappa$ for which the Riemannian metric tensor associated with the action $\beta^{\kappa}$ in equation \eqref{eqn:family of actions of special linear group on states} is actually a quantum monotone metric tensor \cfr{proposition \ref{prop: opmon}}.
%%%%%

The work is structured as follows.
%%%%
In section  \ref{sec:geometry_of_the_space_of_quantum_states_basic_notions} we discuss those differential geometric properties of the manifold of normalized and un-normalized quantum states that are necessary to the proof of our main results.
%%%
In section \ref{sec: main results} we set up the problem and prove our main results for $\Glh$ and $T^{*}\Uh$, namely, proposition \ref{prop: GL(H) actions and monotone quantum metric tensors} and proposition \ref{prop: T*U(H) actions and monotone quantum metric tensors}.
%%%%%%
In section \ref{sec: conclusions} we discuss our results and some  possible future  directions of investigation.
%%%%%

%%%%%%%%%%%%%%%%%%%%%%%%%%%%%%
%%%%%%%%%%%%%%%%%%%%%%%%%%%%%%
%%%%%%%%%%%%%%%%%%%%%%%%%%%%%%
%%%%%%%%%%%%%%%%%%%%%%%%%%%%%%
%%%%%%%%%%%%%%%%%%%%%%%%%%%%%%
\section{Geometry of (un-normalized) quantum states} \label{sec:geometry_of_the_space_of_quantum_states_basic_notions}

%%%%%%%%%%%%%%%%%%%%%%%%%%%%%%
%%%%%%%%%%%%%%%%%%%%%%%%%%%%%%
In this section, the construction of the space of quantum states \cite{A-S-1999, B-Z-2006} is briefly described and some of its geometric features are recalled, this gives the setting for our discussion. Then we give a hint to the role played by group actions in the context of Quantum Mechanics and introduce some particular group actions that will be needed in order to get to the main result of this work. Finally, the concept of \emph{monotone metric}, which is crucial in the context of Quantum Information Geometry, is introduced.
%%%%%%%%%%%%%%%%%%%%%%%%%%%%%%
%%%%%%%%%%%%%%%%%%%%%%%%%%%%%%
%%%%%%%%%%%%%%%%%%%%%%%%%%%%%%
%%%%%%%%%%%%%%%%%%%%%%%%%%%%%%
%%%%%%%%%%%%%%%%%%%%%%%%%%%%%%
%%%%%%%%%%%%%%%%%%%%%%%%%%%%%%
%%%%%%%%%%%%%%%%%%%%%%%%%%%%%%

In standard quantum mechanics \cite{Dirac-1958,von-Neumann-1955}, a quantum system is mathematically described with the aid of a complex Hilbert space $\hh$.
%%%%
The bounded observables of the system are identified with the self-adjoint elements in the algebra  $\bh$ of bounded linear operators   $\hh$, and the set of all such elements is denote with $\mathcal{B}_{sa} (\mathcal{H})$.
%%%%
The physical states of the system are identified with the so-called \grit{density operators} on $\hh$.
%%%%
In order to define what a density operator is, we start recalling that  $\omega\in\bhsa$ is said to be positive semi-definite if
%%%%%%%%%%%%%%%%%%%%%%%%%%%%%%
\begin{equation}\label{eqn: positive semidefinite operators}
	\bra{\psi} \omega \ket{\psi} \ge 0  \quad \forall  \psi  \in \hh ,
\end{equation}
and its customary to write  $\omega \geq 0$ and $\omega >0$ when it is also invertible.
%%%%%%%%%%%%%%%%%%%%%%%%%%%%%%
%It turns out equation \eqref{eqn: positive semidefinite operators}   is equivalent to the existence of a unique $a\in\mathcal{B}_{sa}(\hh)$ such that $\omega = a^2$ \cite{Blackadar-2006}.
%%%%%
The space of positive semi-definite operators is denoted by $\overline{ \mathcal{ P(H)}}$ so that
\begin{equation}
	\overline{ \mathcal{ P(H)}} = \{  \omega \in \bh| \, \, \omega \geq 0\} ,
\end{equation}
and its elements may be referred to as \grit{un-normalized quantum states} for reasons that are clarified below.
%%%%%%%%%%%%%%%%%%%%%%%%%%%%%%
From a geometrical point of view, $\overline{ \mathcal{ P(H)}}$ is a convex cone.
%%%%%%
A \emph{density operator} $\rho$ is just an element in  $\overline{ \mathcal{ P(H)}}$ satisfying the normalization condition $\Tr \rho = 1$.
%%%%%
This linear condition defines a hyperplane
%%%%%%%%%%%%%%%%%%%%%%%%%%%%%% 
\begin{equation}\label{eqn: trace-1 hyperplane} 
	\mathcal{B}^1_{sa} ( \hh) = \{ a \in \mathcal{ B}_{sa} ( \hh) | \,\, a = a^\dagger, \,\, \Tr \, a = 1   \}
\end{equation}
%%%%%%%%%%%%%%%%%%%%%%%%%%%%%%
in $\mathcal{ B}_{sa} ( \hh)$.
%%%%
As anticipated before, physical states are identified with density operators and thus the space of quantum states reads
%%%%%%%%%%%%%%%%%%%%%%%%%%%%%%
\begin{equation}\label{eqn: space of quantum states}
	\overline{ \stsph} = \{ \rho \in \overline{ \mathcal{ P(H)}}, \, \, \Tr\, \rho = 1 \} ,
\end{equation}
and thus the nomenclature “un-normalized quantum states” for elements in $\overline{\mathcal{P}(\hh)}$ appears justified.
%%%%%%%%%%%%%%%%%%%%%%%%%%%%%%
Clearly, $\overline{ \stsph}$ is  given by the intersection between the convex cone $\overline{ \mathcal{ P(H)}}$   and the hyperplane $\mathcal{ B}^1_{sa}$ and thus is a convex set.
%%%%%

\begin{remark}\label{rem: classical and quantum}
It is worth noting that there is a very deep analogy between the space $\overline{ \mathcal{ P(H)}}$ of positive semi-definite operators and the  space $\mathscr{M}_{\mu}(M) $ of classical measures on the measurable space $M$ which are absolutely continuous with respect to the reference measure $\mu$.
%%%%
Indeed, the elements in these sets are defined in terms of a suitable positivity condition, and they also act on suitable algebras: positive semi-definite operators act on $\bh$ through $\mathrm{Tr}(\omega a)$, while measures on $M$ acts on the algebra $\mathcal{L}^{\infty}(M,\mu)$ through integration.
%%%%
This analogy finds the perfect mathematical formalization  in the context of the theory of $C^{*}/W^{*}$-algebras \cite{Blackadar-2006,B-R-1987-1,Landsman-2017,Takesaki-2002}, where it turns out that both $\overline{ \mathcal{ P(H)}}$ and $\mathscr{M}_{\mu}(M) $ arise as the space of normal positive linear functionals  on  suitable  $C^{*}/W^{*}$-algebras: the algebra $\bh$  in the quantum case, and the algebra $\mathcal{L}^{\infty}(M,\mu)$ in the classical case.
%%%%

Then, recalling that the space of quantum states $\overline{ \stsph}$ arise as a normalized version of $\overline{ \mathcal{ P(H)}}$ ({\itshape cfr.} equation \eqref{eqn: space of quantum states}), it is almost immediate to spot the analogy between  and the space  $\mathscr{P}_{\mu}(M) $ of probability measures on $M$ which are absolutely continuous with respect to the reference measure $\mu$.
%%%%
Of course, this analogy too finds a mathematical formalization  in the context of the theory of $C^{*}/W^{*}$-algebras, where now both sets appear as the space of (functional analytic) states on the aforementioned algebras.
%%%%%

This parallel is being exploited to give a unified account of some aspects of classical and quantum information geometry   \cite{C-J-S-2020-02,C-J-S-2020,Naudts-2018,Naudts-2022,Naudts-2021}.
%%%%%
\end{remark}

In the rest of this work, as it is often done in the context of quantum information theory \cite{N-C-2011}, we restrict our attention to the finite-dimensional case in which $\hh$ has complex dimension $n<\infty$.
%%%%%
Then, it is proved that both $\overline{ \mathcal{ P(H)}}$ and $\overline{ \stsph}$ may be endowed with the structure of stratified manifold whose underlying topological structure coincides with the topology inherited from $\bh$ \cite{DA-F-2021}.
%%%%
It turns out that the  \emph{strata} of these stratified manifolds can be described in terms of a particular action of the general linear group $\Glh$ \cite{C-DC-I-L-M-2017,C-I-J-M-2019,C-J-S-2020,G-K-M-2005,G-K-M-2006}.
%%%%%
Specifically, $\Glh$ acts on the whole $\bh$ according to 
\be\label{eqn: linear action of GL(H)}
(\gr,a)\mapsto \hat{\beta}(\gr,a)=\gr\,a\,\gr^{\dagger} .
\ee
%%%%
It is important to note that, when we restrict the action $\hat{\beta}$ in such a way that it acts only on positive elements and only by means of elements in $\Uh$, it reduces to the canonical action
%%%%%%%%%%%%%%%%%%%%%%%%%%%%%% 
\begin{equation} \label{eqn:definition_alpha}
	( U, \omega) \mapsto \hat{\alpha}(U,\omega) = U \omega U^*
\end{equation}
%%%%%%%%%%%%%%%%%%%%%%%%%%%%%%
of the unitary group.
%%%%%%
The action $\hat{\beta}$ is linear, and it is a matter of direct inspection to check that it preserves both $\bhsa$ and $\overline{ \mathcal{ P(H)}}$ and that the orbits through   $\overline{ \mathcal{ P(H)}}$  are made of positive semi-definite operators of the same rank, denoted by $   \mathcal{ P}^{k}(\hh) $ where $k\leq n $ is the rank.
%%%%
These orbits thus become homogeneous manifold and their underlying smooth structures agree with those associated with the stratification of $\overline{ \mathcal{ P(H)}}$ \cite{DA-F-2021}.
%%%%%

In particular, we are interested in the  maximal stratum 
\be
\mathcal{ P}(\hh) \equiv    \mathcal{ P}^{n}(\hh) = \{ \omega\in \overline{ \mathcal{ P(H)}}, \, \, \omega >0 \},
\ee
i.e., the space of invertible elements in   $\overline{ \mathcal{ P(H)}}$, which  forms the open interior of  $\overline{ \mathcal{ P(H)}}$.
%%%%%
The tangent space $T_\omega \mathcal{ P(H)}$ of $\mathcal{P(H)}$ at $\omega \in \mathcal{ P(H)}$ is isomorphic to $\mathcal{ B}_{sa}( \hh)$, since $\mathcal{ P ( H)}$ is an open set in $\mathcal{ B}_{sa}( \hh)$.
%%%%%%%%%%%%%%%%%%%%%%%%%%%%%%
Since $\mathcal{P}(\hh)$ is a homogeneous manifold, the tangent space at each point can be described in terms of the fundamental vector fields of the action $\hat{\beta}$ evaluated at a point $\omega \in \mathcal{ P(H)}$ \cite{A-M-R-2012}.
%%%%%%%%%%%%%%%%%%%%%%%%%%%%%%
Recalling that the Lie algebra $\mathfrak{gl}(\hh)$ of $\Glh$ is essentially $\bh$ endowed with the standard commutator, a curve in the group $\Glh$ can be written as
\begin{equation} \label{eqn:curve_on_GLH}
	g(t) = e^{ \frac{1}{2} t ( \mathbf{a} - i \mathbf{b}) }, 
\end{equation}
with $\mathbf{ a}$ and $\mathbf{b}$ self-adjoint operators.
%%%%%
Therefore, the fundamental vector field associated with $ \mathbf{a} - i \mathbf{b}$ at the point $\omega$ reads
\begin{equation} \label{eqn:fundamental_vector_fields_beta}
	\hat{Z}_{\mathbf{ a} \, \mathbf{ b}} (\omega) = \frac{d}{d t} \beta( g(t), \omega) \big|_{t = 0} = [ \omega, \mathbf{b}] + \{ \omega, \mathbf{a}\} \equiv \hat{X}_\mathbf{ b} (\omega) + \hat{Y}_\mathbf{ a} (\omega),
\end{equation}
%%%%%%%%%%%%%%%%%%%%%%%%%%%%%%
where we have used the notation
%%%%%%%%%%%%%%%%%%%%%%%%%%%%%%
\begin{equation}\label{eqn: commutators and anticommutators}
	%%%%%%%%%%%%%%%%%%%%%%%%%%%%%%
	\begin{split} 	
		[  \mathbf{a}, \mathbf{b} ]  & = \frac{i}{2} ( \mathbf{a} \mathbf{b} - \mathbf{b} \mathbf{a}),\\
		\{ \mathbf{a}, \mathbf{b} \} & = \frac{1}{2} ( \mathbf{a} \mathbf{b} + \mathbf{b} \mathbf{a}),
	\end{split}
	%%%%%%%%%%%%%%%%%%%%%%%%%%%%%%
\end{equation}
%%%%%%%%%%%%%%%%%%%%%%%%%%%%%%
and we have set
%%%%%%%%%%%%%%%%%%%%%%%%%%%%%%
\begin{equation}
	%%%%%%%%%%%%%%%%%%%%%%%%%%%%%%
	\begin{split} \label{eqn:definition_X_and_Y_PH}
		\hat{X}_{\mathbf{ b}} ( \omega)  & := \hat{Z}_{\mathbf{ 0} \, \mathbf{ b}} (\omega) = [  \omega, \mathbf{b}  ],
	\end{split}
	%%%%%%%%%%%%%%%%%%%%%%%%%%%%%%
\end{equation}
and
\begin{equation}
	%%%%%%%%%%%%%%%%%%%%%%%%%%%%%%
	\begin{split} \label{eqn:definition_Y_PH}
		\hat{Y}_{\mathbf{ a}} ( \omega)  & := \hat{Z}_{\mathbf{ a} \, \mathbf{ 0}} (\omega) = \{ \omega, \mathbf{a}  \} .
	\end{split}
	%%%%%%%%%%%%%%%%%%%%%%%%%%%%%%
\end{equation}
%%%%%%%%%%%%%%%%%%%%%%%%%%%%%%
As mentioned before, when we restrict $\hat{\beta}$  to $\Uh$ we obtain the canonical action $\hat{\alpha}$ of $\Uh$.
%%%%
Therefore, since the Lie algebra $\uh$ is just the space of skew-adjoint elements in $\bh=\mathfrak{gl}(\hh)$, setting $\mathbf{a} = 0$ in equation \eqref{eqn:curve_on_GLH}, we immediately obtain that the fundamental vector fields of $\hat{\alpha}$ are recovered as the fundamental vector fields $\hat{X}_{\mathbf{b}} = \hat{Z}_{\mathbf{0} \, \mathbf{ b}}$ of the action $\hat{\beta}$.
%%%%
Concerning the vector fields of the form $\hat{Y}_{a}$, taking $a=\mathbb{I}$, we obtain the vector field 
\be\label{eqn: dilation vector field}
\Delta(\omega):=\hat{Y}_{\mathbb{I}}(\omega)=\omega
\ee
which represents the infinitesimal generator of the Lie group $\mathbb{R}_{+}$ acting on $\mathcal{P}(\hh)$ by dilation.
%%%%
However, in general,  it turns out that $[\hat{Y}_{a},\hat{Y}_{b}]= \hat{X}_{[b,a]}$ \cite{C-DC-I-L-M-2017,C-J-S-2020} so that they do not form a Lie subalgebra.
%%%

Besides $\Glh$, also the cotangent Lie group $T^{*}\Uh$ acts on $\mathcal{P}(\hh)$ in such a way that the latter becomes a homogeneous manifold of $T^{*}\Uh$.
%%%%% 
Specifically, the action is given by
\begin{equation} \label{eqn: BKM-action on plf}
( (U, \mathbf{a}), \omega) \mapsto \hat{\gamma} (( U, \mathbf{ a}), \rho) =  \mathrm{e}^{U \log \rho \, U^\dagger + \mathbf{a}} , 
\end{equation}
where we used the canonical identifications $T^{*}\Uh\cong \Uh \times\uh^{*}\cong \Uh\times \bhsa$.
%%%%
It is not hard to check that, if we restrict to $\Uh$ by considering only elements of the type $(U,0)$, the action $\hat{\gamma}$ reduces to the action $\hat{\alpha}$ of $\Uh$ on $\mathcal{P}(\hh)$.
%%%%%%
The action $\hat{\gamma}$ is smooth with respect to the previously mentioned smooth structure on $\mathcal{P}(\hh)$ associated with the action $\hat{\beta}$, and it is a transitive action.
%%%%
Therefore, we conclude that the smooth structure underlying $\mathcal{P}(\hh)$ when thought of as a homogeneous manifold for $T^{*}\Uh$ coincides with the smooth structure underlying $\mathcal{P}(\hh)$ thought of as a homogeneous manifold for $\Glh$.
%%%%%
The action $\hat{\gamma}$ is basically related with the isomorphism $\Phi\colon\mathcal{P}(\hh)\ra \bhsa$ given by $\omega\ra\Phi(\omega):=\ln(\omega)$ and its inverse $\Phi^{-1}\colon\bhsa\ra\mathcal{P}(\hh) $ given by $a\ra\Phi^{-1}(x):=\mathrm{e}^{x}$.
%%%%
Indeed, its clear that $T^{*}\Uh\cong \Uh \times\uh^{*}\cong \Uh\times \bhsa$ acts on $\bhsa$ through
\be
x\mapsto \zeta((U,\mathbf{a}),x)=U\mathbf{x}U^{\dagger} + \mathbf{a},
\ee
and it is a matter of direct inspection to show that 
\be 
\hat{\gamma}_{(U,\mathbf{a})}=\Phi^{-1}\circ \zeta_{(U,\mathbf{a})}\circ\Phi .
\ee
%%%%
Thinking to $\Uh$ as a subgroup of the rotation group of the vector space $\bhsa$, it follows that the action $\zeta$ coincides with the restriction to $T^{*}\Uh$ of the standard action of the affine group on $\bhsa$.
%%%%%
The action $\hat{\gamma}$ can not be extended to the whole $\overline{\mathcal{P}(\hh)}$ essentially because $\Phi$ and its inverse can not be extended.
%%%%%%
Concerning the fundamental vector fields $\hat{W}_{\mathbf{ a} \, \mathbf{ b}}$ of $\hat{\gamma}$, we have that $\hat{W}_{\mathbf{ 0} \, \mathbf{ b}} =\hat{X}_{\mathbf{b}}$ as in equation \eqref{eqn:definition_X_and_Y_PH}, while
\be\label{eqn: gradient vector fields of BKM on P(H)}
\hat{W}_{\mathbf{a}\,\mathbf{0}}(\omega)\,=\,\frac{\mathrm{d}}{\mathrm{d}t}\left(\mathrm{e}^{ \ln(\omega) + t\mathbf{a}}\right)_{t=0}= \int_{0}^{1}\,\mathrm{d}\lambda\,\left(\omega^{\lambda}\,\mathbf{a} \,\omega^{1-\lambda}\right)\,,
\ee
where we used the well-known equality 
\be\label{eqn: derivative of exponential of operator}
\frac{\mathrm{d}}{\mathrm{d}t}\, \mathrm{e}^{A(t)} \,=\,\int_{0}^{1}\,\mathrm{d}\lambda\,\left(\mathrm{e}^{\lambda\,A(t)}\,\frac{\mathrm{d}}{\mathrm{d}t}(A(t))\,\mathrm{e}^{(1-\lambda)A(t)}\right)\,,
\ee
which is valid for every smooth curve $A(t)$ inside $\bh$ (remember that the canonical immersion of $\posh$ inside $\bh$ is smooth) \cite{Suzuki-1997}.
%%%
 
%%%%%%

\vsp

We turn now our attention to faithful quantum states.
%%%%%
The action in equation \eqref{eqn: linear action of GL(H)}  does not preserve the hyperplane $\mathcal{B}^1_{sa} ( \hh)$ in equation \eqref{eqn: trace-1 hyperplane}, and thus it also does not preserve $\overline{\stsph}$.
%%%%
However, as already anticipated in equation \eqref{eqn:action of general linear group on states BH}, it is possible to suitably renormalize $\hat{\beta}$ to obtain the action 
\begin{equation}\label{eqn:action of general linear group on states BH}
	(\gr,\rho)\mapsto \beta(\gr,\rho):= \frac{\gr \, \rho \, \gr^{\dagger}}{\mathrm{Tr}(\gr \, \rho \, \gr^{\dagger})}.
\end{equation}
%%%%
The normalization is recovered at the expense of the linearity/convexity of the action.
%%%%%
However, when we  restrict to the unitary group $\Uh$, the action $\beta$ reduces to the canonical action
%%%%%%%%%%%%%%%%%%%%%%%%%%%%%%
\begin{equation} \label{eqn:definition_alpha_stsp}
 ( U, \rho) \mapsto \alpha ( U, \rho) = U \rho U^*  
\end{equation}
%%%%%%%%%%%%%%%%%%%%%%%%%%%%%%
of the unitary group on the space of states which does preserve convexity.
%%%%%%
Analogously to what happens for the action $\hat{\beta}$ on $\overline{\mathcal{P}(\hh)}$, the orbits of $\beta$ are made up of quantum states with the same fixed rank, and any such orbit is denoted as $\stsp^k(\mathcal{H})$ where $k$ is the rank.
%%%%%
These orbits thus become homogeneous manifold and their underlying smooth structures agree with those associated with the stratification of $\overline{\stsph}$ \cite{DA-F-2021}.
%%%%
Moreover, each manifold $\stsp^k(\mathcal{H})$ can be seen as a submanifold of $\mathcal{P}^{k}(\hh)$ singled out by the intersection with the affine hyperplane $\mathcal{ B}^1_{sa}( \hh)$.
%%%%%
%%%%%
It is worth mentioning that the partition of  $\overline{\stsph}$ in terms of manifold of quantum states of fixed rank was also exploited in \cite{A-N-2000,Fujiwara-1999}, however, as far as the authors know, the homogeneous manifold structures was firstly understood in \cite{G-K-M-2005,G-K-M-2006} and the stratified structure in \cite{DA-F-2021}.
%%%%%%%

\begin{remark} \label{rem:analogy_simplex_space_of_quantum_states}
Building on remark \ref{rem: classical and quantum}, for a reader familiar with Classical Information Geometry it may be useful to think of the space of quantum states of an $n$-level quantum system as the quantum analogue of the $(n-1)$-simplex, with the \emph{strata} of the space of quantum states taking the place of the \emph{faces} of the simplex. 
%%%%
A thorough discussion of this analogy can be found in \cite{C-DC-DN-V-2022,C-DC-L-M-M-V-V-2018,M-M-V-V-2017}.
%%%%%%
\end{remark}
%%%%%%%%%%%%%%%%%%%%%%%%%%%%%%

%%%%
In particular, we focus on the stratum of maximal rank, i.e. invertible, or \emph{faithful} states  
\be
\stsp^n(\mathcal{ H})\equiv\stsph = \{ \rho \in \mathcal{P}(\hh)\,| \, \,   \Tr\, \rho = 1 \} .
\ee 
%%%%%%%%%%%%%%%%%%%%%%%%%%%%%%
%%%%%%%%%%%%%%%%%%%%%%%%%%%%%%
The tangent space $T_\rho \stsph$ of $\stsp(\hh)$ at $\rho$ is given by self-adjoint operators  with the additional property of being traceless, i.e., we have
%%%%%%%%%%%%%%%%%%%%%%%%%%%%%%
\be\label{eqn: tangent space at a state}
	T_\rho \stsph \cong \mathcal{B}^0_{sa} ( \hh) = \{ a \in \mathcal{ B}_{sa} ( \hh) | \,\, a = a^\dagger, \,\, \Tr \, a = 0   \}.
\ee
%%%%%
Since $\stsph$ is a homogeneous manifold, the tangent space can be described using the fundamental vector fields of the action $\beta$ following what is done for   $\hat{\beta}$.
%%%%
The fundamental vector fields of the action $\beta $ evaluated at a point $\rho \in \stsph$ are given by
%%%%%%%%%%%%%%%%%%%%%%%%%%%%%%
\begin{equation}
	Z_{\mathbf{ a} \, \mathbf{ b}} (\rho) = \frac{d}{d t} \beta( g(t), \rho) \big|_{t = 0} = [ \rho, \mathbf{b}] + \{ \rho, \mathbf{a}\} - \rho \Tr \left( \{ \rho, \mathbf{a}\} \right) \equiv X_\mathbf{ b} (\rho) + Y_\mathbf{ a} (\rho).
\end{equation}
Where $g(t)$ is defined as in \eqref{eqn:curve_on_GLH} and now we have set
%%%%%%%%%%%%%%%%%%%%%%%%%%%%%%
\begin{equation} \label{eqn:definition_X_and_Y_stsp}
	%%%%%%%%%%%%%%%%%%%%%%%%%%%%%%
	\begin{split} 
		X_{\mathbf{ b}} ( \rho)  & := Z_{\mathbf{ 0} \, \mathbf{ b}} (\rho) = [ \rho, \mathbf{b}],
	\end{split}
	%%%%%%%%%%%%%%%%%%%%%%%%%%%%%%
\end{equation} 
and
\begin{equation} \label{eqn:definition_Y_stsp}
	%%%%%%%%%%%%%%%%%%%%%%%%%%%%%%
	\begin{split} 
		Y_{\mathbf{ a}} ( \rho)  & := Z_{\mathbf{ a} \, \mathbf{ 0}} (\rho) = \{ \rho, \mathbf{a} \} - \rho \Tr \left( \{ \rho, \mathbf{a}\} \right) .
	\end{split}
	%%%%%%%%%%%%%%%%%%%%%%%%%%%%%%
\end{equation}
%%%%%%%%%%%%%%%%%%%%%%%%%%%%%%
Again in analogy with what happens on $\mathcal{P(H)}$, the fundamental vector fields of the action $\alpha$ of $\Uh$ are identified with the  vector fields $X_{\mathbf{ b}}$. 
%%%%

As anticipated in the introduction, the cotangent Lie group $T^{*}\Uh$ also acts on $\stsph$ through the action $\gamma$ given in equation \eqref{eqn:action of general linear group on states BKM}.
%%%%%
This action is smooth with respect to the previously mentioned smooth structure on $\stsph$ associated with the action $\beta$, and it is a transitive action.
%%%%
Therefore, we conclude that the smooth structure underlying $\stsph$ when thought of as a homogeneous manifold for $T^{*}\Uh$ coincides with the smooth structure underlying $\stsph$ thought of as a homogeneous manifold for $\Glh$.
%%%%
Moreover, when restricting to $\Uh$, a direct computation shows that the action $\gamma$ reduces to the standard action $\alpha$ of $\Uh$ on $\stsph$.
%%%%
%
%From equation \eqref{eqn:action of general linear group on states BKM} and equation \eqref{eqn: BKM-action on plf} it follows that the action $\gamma$ is a kind of normalized version of the action $\hat{\gamma}$.
%%%%%
%%%%%
The fundamental vector fields $W_{\mathbf{a}\,\mathbf{b}}$ of $\gamma$ are then easily found.
%%%%
In particular, $W_{\mathbf{0}\,\mathbf{b}}=X_{\mathbf{b}}$ as in equation \eqref{eqn:definition_X_and_Y_stsp}, and  
%%%%%%%
\be\label{eqn: gradient vector fields of BKM on S(H)}
W_{\mathbf{a}\,\mathbf{0}} (\rho)\,=\,\frac{\mathrm{d}}{\mathrm{d}t}\left(\frac{\mathrm{e}^{ \ln(\rho)  + t\mathbf{a}}}{\Tr\left(\mathrm{e}^{ \ln(\rho)  + t\mathbf{a}}\right)}\right)_{t=0}=\int_{0}^{1}\,\mathrm{d}\lambda\,\left(\rho^{\lambda}\,\mathbf{a} \,\rho^{1-\lambda}\right) - \Tr(\rho\,\mathbf{a})\,\rho\,,
\ee
where we again exploited equation \eqref{eqn: derivative of exponential of operator} (remember that the canonical immersion of $\stsph$ inside $\bh$ is smooth).
%%%%%%
Concerning the vector fields in equation \eqref{eqn: gradient vector fields of BKM on S(H)}, it is worth mentioning that they already appeared in \cite{Petz-1994} in connection with the Bogoliubov-Kubo-Mori metric tensor, and then in the recent work \cite{A-L-2021} where the finite transformations they induce are exploited in the definition of a Hilbert space structure on $\stsph$ which is the quantum counterpart of a classical structure relevant in estimation theory.
%%%%
However, as far as the author know, the group-theoretical aspects relating the vector fields in equation \eqref{eqn: gradient vector fields of BKM on S(H)} with the action $\gamma$ of $T^{*}\Uh$ were first investigated in \cite{Ciaglia-2020}.
%%%%%%

Despite the lack of a universally recognized physical interpretation for un-normalized quantum states in $\overline{\mathcal{P}(\hh)}$, it turns out that they provide a more flexible environment in which to perform the mathematics needed to prove the main result of this work.
%%%%
Intuitively speaking, it is already clear from the very definition of the actions $\hat{\beta}$ and $\beta$  that imposing the linear normalization constraint needed to pass to (normalized) quantum states leads to emergence of non-linear aspects which destroy the inherent convexity of the space of quantum states.
%%%%
Moreover, following the ideology expressed in \cite{A-J-L-S-2017}, it can also be argued that the choice of a normalization has a somewhat arbitrary flavour that does not really encodes physical information because basically nothing really serious happens if we decide to normalize to $\pi^{2}$ rather than to $1$.
%%%%
Accordingly, the choice to work with un-normalized states first and then appropriately “project the results” to normalized quantum states appears to us as not entirely unphysical.
%%%%%

Following this line of thought, we will always work on $\mathcal{P}(\hh)$ making sure that all the structure and results may be appropriately “projected” to $\stsph$.
%%%%
At this purpose, it is relevant to introduce a projection map from $\pi\colon\mathcal{P} ( \mathcal{H})\ra\stsph$ as
\begin{equation}\label{eqn: projection map}
  \omega \mapsto \pi(\omega) = \frac{\omega}{\Tr( \omega)} ,
\end{equation}
and an associated section given by the  natural immersion map $i\colon\stsph \ra\mathcal{ P ( H)}$ reading
\be\label{eqn: immersion map}
\rho\mapsto i(\rho)=\rho .
\ee 
%%%%%%%%%%%%%%%%%%%%%%%%%%%%%%
It is not hard to show that $i$ is an embedding, while $\pi$ is a surjective submersion.
%%%%%
Moreover, it is also possible to “extend” these maps to the whole $\overline{\mathcal{P}(\hh)}$ and $\overline{\stsph}$ in the obvious way, thus obtaining a continuous projection map and a continuous immersion map that preserve the stratification of $\overline{\mathcal{P}(\hh)}$ and $\overline{\stsph}$ and are smooth on each strata.
%%%%%

As mentioned before, bounded physical \emph{observables} are described by means of self-adjoint operators in $\bhsa$. 
%%%
Then, to any observable $\mathbf{a}$, it is possible to associate a smooth function $\hat{l}_\mathbf{a} ( \omega): \mathcal{ P ( H)} \ra\mathbb{R}$ given by
%%%%%%%%%%%%%%%%%%%%%%%%%%%%%%
\begin{equation}\label{eqn:expectation_value_functions_cone}
	  \hat{l}_\mathbf{a} ( \omega) := \Tr{\, (\omega \, \mathbf{a})} ,
\end{equation}
%%%%%%%%%%%%%%%%%%%%%%%%%%%%%%
this is referred to as \emph{expectation value function} of the observable $\mathbf{a}$. 
%%%%
Of course, expectation value functions can also be defined  on the space of quantum states $\stsph$ setting
%%%%%%%%%%%%%%%%%%%%%%%%%%%%%%
\begin{equation}\label{eqn:expectation_value_functions_stsp}
	l_\mathbf{a} ( \rho) := \Tr{\, (\rho \, \mathbf{a})} \in \mathbb{R}, 
\end{equation}
and it turns out that $\hat{l}_{a}$ is connected to $l_{a}$ by means of the pull-back with respect to $i$, i.e., it holds
\begin{equation}
	l_\mathbf{a} = i_{\stsp}^* \hat{l}_\mathbf{a}.
\end{equation}
%%%%
By relaxing smoothness to continuity, it is possible to extend the expectation value functions to the whole $\overline{\mathcal{P}(\hh)}$ and the whole $\overline{\stsph}$.
%%%%%

It is a matter of direct calculation using the very definition of fundamental vector fields for both $\hat{\alpha}$ and $\alpha$ \cfr{equation \eqref{eqn:definition_alpha}  and equation \eqref{eqn:definition_alpha_stsp}} to show that
\be\label{eqn: fundamental vector fields of U(H) on expectation value functions}
\begin{split}
\mathcal{L}_{\hat{X}_{\mathbf{b}}}\hat{l}_{\mathbf{a}}=\hat{l}_{[\mathbf{b},\mathbf{a}]} \\
\mathcal{L}_{ X_{\mathbf{b}}}l_{\mathbf{a}}=l_{[\mathbf{b},\mathbf{a}]},
\end{split}
\ee
where $[\cdot,\cdot]$ is as in equation \eqref{eqn: commutators and anticommutators}.
%%%%%%%%%%%%%%%%%%%%%%%%%%%%

By direct computation, it is possible to spot an interesting intertwine between the maps $\pi$ and $i$ and the actions $\hat{\beta}$ and $\beta$ given by
\begin{equation} \label{eqn: intertwine between actions GL(H)}
	\beta  = \pi \circ \hat{\beta}  \circ (\mathrm{Id}_{\Glh}\times i) 
\end{equation}
where $\mathrm{Id}_{\Glh}$ is the identity map on $\Glh$.
%%%%
Analogously, we obtain
\begin{equation} \label{eqn: intertwine between actions BKM}
	\gamma  = \pi \circ \hat{\gamma}  \circ (\mathrm{Id}_{T^{*}\Uh}\times i) 
\end{equation}
where $\mathrm{Id}_{T^{*}\Uh}$ is the identity map on $T^{*}\Uh$.
%%%%
Equation \eqref{eqn: intertwine between actions GL(H)} and equation \eqref{eqn: intertwine between actions BKM} explain in which sense $\beta $ and $\gamma$ are a kind of normalized version of the actions $\hat{\beta}$ and $\hat{\gamma}$ respectively.
%%%%%%%%%
The immersion map $i$ also allows us to obtain a pointwise relation between the fundamental vector fields $\hat{X}_{\mathbf{ b}}$ and $X_{\mathbf{ b}}$ and between the fundamental vector fields $\hat{Y}_{\mathbf{ a}}$ and $Y_{\mathbf{ a}}$ in terms of the tangent map $T_{\rho}i$ to $i$ at $\rho$.
%%%%%
Indeed,  from equation \eqref{eqn:definition_X_and_Y_PH}, equation \eqref{eqn:definition_Y_PH}, equation \eqref{eqn: dilation vector field}, equation \eqref{eqn:definition_X_and_Y_stsp}, and equation \eqref{eqn:definition_Y_stsp}, it follows that 
%%%%%%%%%%%%%%%%%%%%%%%%%%%%%%
\be 
\begin{split}
T_\rho i \left( X_\mathbf{ b} ( \rho) \right) & =	\hat{X}_\mathbf{ b} ( i (\rho)) \\
T_{\rho}i \left( Y_{\mathbf{ a}} (\rho)\right)&=\hat{Y}_{\mathbf{a}} \left( i(\rho) \right) -  \Tr \left(\rho\, \mathbf{a}  \right)\, \Delta(i(\rho)) .
\end{split}
\end{equation}
%%%%%
Accordingly, we conclude that $X_{\mathbf{b}}$ is $i$-related with $\hat{X}_{\mathbf{b}}$ while $Y_{\mathbf{a}}$ is $i$-related with $\hat{Y}_{\mathbf{a}} - \hat{l}_{\mathbf{a}}\Delta$.
%%%%%%
Analogously, from equation \ref{eqn: gradient vector fields of BKM on P(H)}, equation \eqref{eqn: dilation vector field}, and equation \ref{eqn: gradient vector fields of BKM on S(H)}, it follows
that
\be
T_{\rho}i(W_{\mathbf{a}\,\mathbf{0}} (\rho))= \hat{W}_{\mathbf{a}\,\mathbf{0}} (\rho)-  \Tr \left(\rho \, \mathbf{a}  \right)\, \Delta(i(\rho)) ,
\ee
which means that $W_{\mathbf{a}\,\mathbf{0}} $ is $i$-related with $\hat{W}_{\mathbf{a}\,\mathbf{0}}  - \hat{l}_{\mathbf{a}}\Delta$.
%%%%%%%%%%%%%%%%%%%%%%%%%%%%%%
%%%%%%%%%%%%%%%%%%%%%%%%%%%%%%
%%%%%%%%%%%%%%%%%%%%%%%%%%%%%%
%%%%%%%%%%%%%%%%%%%%%%%%%%%%%%
%%%%%%%%%%%%%%%%%%%%%%%%%%%%%%

%%%%%%%%%%%%%%%%%%%%%%%%%%%%%%
%%%%%%%%%%%%%%%%%%%%%%%%%%%%%%
%%%%%%%%%%%%%%%%%%%%%%%%%%%%%%
%%%%%%%%%%%%%%%%%%%%%%%%%%%%%%
%%%%%%%%%%%%%%%%%%%%%%%%%%%%%%
%%%%%%%%%%%%%%%%%%%%%%%%%%%%%%
\subsection{Quantum monotone metric tensors} \label{subsec:geometry_of_the_space_of_quantum_states_information_metrics}
%%%%%%%%%%%%%%%%%%%%%%%%%%%%%%
%%%%%%%%%%%%%%%%%%%%%%%%%%%%%%

In the classical case, the Riemannian aspects of most of the manifolds of probability employed in statistics, inference theory, information theory, and information geometry are essentially encoded in a single metric tensor\footnote{We are here deliberately “ignoring” all those Wasserstein-type metric tensors simply because their very definition depends on the existence of additional structures on the sample space.}, namely, the Fisher-Rao metric tensor \cite{Fisher-1922,Mahalanobis-1936,Rao-1945}.
%%%%%
In the case of finite sample spaces, Cencov's pioneering work \cite{Cencov-1982} investigated the Fisher-Rao metric tensor from a  category-theoretic perspective and uncovered the uniqueness of this metric tensor when some invariance conditions are required.
%%%%
Specifically, let $\overline{\mathsf{S}_{n}}$ denote the n-dimensional simplex in $\mathbb{R}^{n}$, i.e., the space of probability distributions on a discrete sample space with n elements, and let $\mathsf{S}_{n}$ denote the interior of $\overline{\mathsf{S}_{n}}$, the space of probability distributions with full support.
%%%%
Note that $\mathsf{S}_{n}$ is a smooth, $(n-1)$-dimensional manifold while $\overline{\mathsf{S}_{n}}$ is a smooth manifold with corners.
%%%%%

%%%%%
A linear map $F\colon\mathbb{R}^{n}\ra\mathbb{R}^{m}$ is called a \grit{Markov morphism} if $F(\overline{\mathsf{S}_{n}})\subset .\overline{\mathsf{S}_{m}}$, and a Markov morphism $F$ is called a \grit{congruent embedding} if $F( \mathsf{S}_{n} )$ is diffeomorphic to $ \mathsf{S}_{n}$.
%%%%%
Congruent embeddings where studied by Cencov who characterized the most general form of these maps \cfr{\cite{Campbell-1986} for yet another characterization of congruent embeddings}.
%%%%%

According to Cencov, the relevant geometrical structures on $\mathsf{S}_{n}$ must all be left unchanged when suitably acted upon by congruent embeddings.
%%%%%
For instance, setting $\mathbb{N}_{>1}=\mathbb{N}\setminus \{\{0\},\{1\}\}$, a family $\{g^{n}\}_{n\in\mathbb{N}_{>1}}$  with $g^{n}$ a smooth Riemannian metric tensor on $\mathsf{S}_{n}$ is called \grit{invariant} if $F^{*}g^{m}=g^{n}$ for every  congruent embedding $F\colon\mathbb{R}^{n}\ra\mathbb{R}^{m}$.
%%%%
Cencov's incredible result was to show that, up to an overall multiplicative positive constant, there is only one invariant family of Riemannian metric tensor for which $g^{n}$ coincides with the Fisher-Rao metric tensor.
%%%%%
Then, much effort has been devoted to extend Cencov's uniqueness result from the case of finite sample spaces to the case of continuous sample spaces leading, for instance, to a formulation on smooth manifolds \cite{B-B-M-2016} and a very general formulation valid for very general parametric models \cite{A-J-L-S-2017}.
%%%%

Finally, it is also worth mentioning Campbell's investigation of Cencov's result when the normalization condition on probability distribution is lifted  \cite{Campbell-1986}.
%%%%
He worked with finite measures with full support on a discrete sample space with $n$ elements and was able to prove that, despite Cencov's uniqueness result is lost,  the non-uniqueness simply amounts to the freedom of vary two smooth functions depending only on the total mass of the measure.
%%%%

\vsp

As already hinted at in remark \ref{rem: classical and quantum} and in remark \ref{rem:analogy_simplex_space_of_quantum_states}, the manifold $\stsph$ may be thought of as the quantum analogue of $\mathsf{S}_{n}$  in the case of finite-level quantum systems.
%%%%%
Then, the quantum analogue of a Markov morphism is a completely-positive and trace-preserving linear (CPTP) map $F\colon\bh\ra\mathcal{B}(\mathcal{K})$ \cfr{\cite{Choi-1975} for the precise definition of CPTP maps and  \cite{Holevo-2001,Holevo-2011} for their role in quantum information}.
%%%%
Quite trivially, a \grit{quantum congruent embedding} could be defined as a CPTP map  $F\colon\bh\ra\mathcal{B}(\mathcal{K})$ such that $F(\stsph)$ is diffeomorphic to $\stsph$.
%%%%%
A typical example of quantum congruent embedding is given by $\alpha_{\mathbf{U}}(\mathbf{a})=\alpha(\mathbf{U},\mathbf{a})=\mathbf{U}\mathbf{a}\mathbf{U}^{\dagger}$.
%%%%%
As far as the authors know, there seems to be no general characterization of these maps at the moment as there is in the classical case.
%%%%

Inspired by Cencov's work, Petz investigated the following problem: characterize the families $\{g^{n}\}_{n\in\mathbb{N}_{>1}}$  with $g^{n}$ a smooth Riemannian metric tensor on $\stsp(\mathbb{C}^{n})$ satisfying the monotonicity property 
\be\label{eqn: monotonicity property riemannian metric tensors}
(g^{n})_{\rho}(v_{\rho},v_{\rho})\geq (g^{m})_{F(\rho)}(T_{\rho}(v_{\rho}),T_{\rho}(v_{\rho}))
\ee
for every CPTP map $F\colon\mathcal{B}(\mathbb{C}^{n})\ra\mathcal{B}(\mathbb{C}^{m})$ and for all $\rho\in\stsp(\mathbb{C}^{n})$.
%%%%%%
He was able to prove \cite{Petz-1996} that, up to an overall multiplicative positive constant, these families of \grit{monotone quantum metric tensors} are completely characterized by operator monotone functions $f\colon \mathbb{R}^{+}\ra\mathbb{R}$  \cite{Loewner-1934} satisfying 
%%%%%%%%%%%%%%%%%%%%%%%%%%%%%%
\be\label{eqn: monotone function}
%%%%%%%%%%%%%%%%%%%%%%%%%%%%%%
\begin{split}
f(t) & = t f( t^{-1}), \\
f(1) & = 1.
\end{split}
%%%%%%%%%%%%%%%%%%%%%%%%%%%%%
\ee
%%%%
In particular, if $\{G^{n}_{f}\}_{n\in\mathbb{N}_{>1}}$  is a family of monotone metric tensors then
\begin{equation} \label{eqn:Petz_characterization 0} 
	\left( G^{n}_{f} \right)_\rho ( \mathbf{v}_{\rho}, \mathbf{w}_{\rho}) = \kappa \Tr_{n} \left(  \mathbf{v}_{\rho} \left( \mathbf{ K}^f_\rho \right)^{-1} ( \mathbf{w}_{\rho})    \right)
\end{equation}
%%%%%%%%%%%%%%%%%%%%%%%%%%%%%%
where $\kappa>0$ is a constant, $\mathbf{v}_{\rho},\mathbf{w}_{\rho}$ are vectors in $T_{\rho} \stsp(\mathbb{C}^{n}) \cong \mathcal{B}^0_{sa}(  \mathbb{C}^{n} ) $,  $\mathbf{K}^f_\rho$ is a superoperator on $\mathcal{B}(  \mathbb{C}^{n} )$ given by
%%%%%%%%%%%%%%%%%%%%%%%%%%%%%%
\begin{equation}\label{eqn:superoperator_K_expression}
	\mathbf{K}^f_{\rho} = f(\mathbf{L}_\rho \mathbf{R}^{-1}_{\rho}) \mathbf{ R}_{\rho} 
\end{equation}
with $f$ the \emph{operator monotone function} mentioned before, and $\mathbf{L}_{\rho}$ and $\mathbf{R}_{\rho}$ are two linear superoperators on $\mathcal{B}(  \mathbb{C}^{n} )$  whose action is given by the left and right multiplication by  $\rho$.
%%%%%%

We briefly mention a recent development toward the use of non-monotone metric tensors in quantum information theory \cite{Suzuki-2021}.
%%%%%

Since every $n$-dimensional complex Hilbert space $\hh$ is isomorphic to $\mathbb{C}^{n}$, we can almost immediately generalize equation \eqref{eqn:Petz_characterization 0}  to define a quantum monotone metric tensor $G_{f}^{\hh}$ on $\stsph$ setting
\begin{equation} \label{eqn:Petz_characterization} 
	\left( G^{\hh}_{f} \right)_\rho ( \mathbf{v}_{\rho}, \mathbf{w}_{\rho}) = \kappa \Tr_{\hh} \left(  \mathbf{v}_{\rho} \left( \mathbf{ K}^f_\rho \right)^{-1} ( \mathbf{w}_{\rho}B)    \right).
\end{equation}
%%%%
In the following, for the sake of notational simplicity, we often simply write $G_{f}$  instead of $G_{f}^{\hh}$ because the Hilbert space $\hh$ is already clear from the context.
%%%%%

If we introduce the operators $\mathbf{e}^{\rho}_{lm}$ diagonalizing $\rho\in\stsph$, that is, such that
\be
\rho=\sum_{j=1}^{n}\,p_{j}^{\rho}\,\mathbf{e}_{jj}^{\rho},
\ee
we can also introduce the super-operators  $E_{kj}^{\rho}$ acting on $\bh$ according to
\be\label{eqn: rho super eigenprojectors}
E_{kj}^{\rho}\left( \mathbf{e}^{\rho}_{lm}\right)\,=\,\delta_{jl}\,\delta_{km}\mathbf{e}_{jk}^{\rho},
\ee
and it is then a matter of straightforward computation to check that
\be\label{eqn: Petz superoperator}
K^{f}_{\rho}=\sum_{j,k=1}^{n}\,p_{k}^{\rho}\,f\left(\frac{p_{j}^{\rho}}{p_{k}^{\rho}}\right)\,E_{kj}^{\rho}
\ee
where $p_{1}^{\rho},...,p_{n}^{\rho}$ are the eigenvalues of $\rho$.
%%%%%%
%Therefore, we also have
%\be\label{eqn: eigendecomposition of Petz superoperator}
%T^{f}_{\rho}=\sum_{j,k=1}^{n}\,\left(p_{k}^{\rho}\,f\left(\frac{p_{j}^{\rho}}{p_{k}^{\rho}}\right)\right)^{-1}\,E_{kj}^{\rho}.
%\ee
%%%% 
Now, whenever $[\mathbf{w}_{\rho},\rho]=0$, from equation \eqref{eqn:Petz_characterization}  and equation \eqref{eqn: Petz superoperator} it follows that 
\be
(G_{f})_{\rho}(\mathbf{v}_{\rho},\mathbf{w}_{\rho})=\sum_{j=1}^{n}\frac{v_{\rho}^{jj} w_{\rho}^{jj}}{p_{j}^{\rho}},
\ee
where $v_{\rho}^{jj}$ and $w_{\rho}^{jj}$ are the diagonal elements of $\mathbf{v}_{\rho}$ and $\mathbf{w}_{\rho}$ with respect to the basis of eigenvectors of $\rho$.
%%%%
It is relevant to note then that in this case we have
\be\label{eqn: Petz and F-R}
(G_{f})_{\rho}(\mathbf{v}_{\rho},\mathbf{w}_{\rho})\,=\,(G_{FR})_{\vec{p}}\,(\vec{a},\vec{b}),
\ee
where $G_{FR}$ is the classical Fisher-Rao metric tensor on $\mathsf{S}_{n}$, and we have set $\vec{p}=(p_{1}^{\rho},...,p_{n}^{\rho})$, $\vec{a}=(v_{\rho}^{11},...,v_{\rho}^{nn})$, and $\vec{b}=(w_{\rho}^{11},...,w_{\rho}^{nn})$.
%%%%
Equation \eqref{eqn: Petz and F-R} holds for every choice of the operator monotone function $f$.
%%%%%%%%%%%%%%%%%%%%%%%%%%%%%%

As mentioned before,  the action $\alpha$ of $\Uh$ in \eqref{eqn:definition_alpha_stsp} gives rise to CPTP maps from $\stsph$ into itself.
%%%%
Moreover, these maps are invertible and their inverses are again CPTP maps from $\stsph$ to itself.
%%%%
Therefore, the monotonicity property in equation \eqref{eqn: monotonicity property riemannian metric tensors} becomes an invariance property, and we conclude that the fundamental vector fields $X_{\mathbf{b}}$ of the action $\alpha$ \cfr{equation \eqref{eqn:definition_X_and_Y_stsp}} are Killing vector fields for every monotone quantum metric tensor $G_{f}$.
%%%%%
Consequently, the unitary group $\Uh$ acts as a sort of universal symmetry group for the metric tensors classified by Petz and thus occupies a prominent role also in the context of Quantum Information Geometry.

To explicitly prove our main results, it is better to work first on $\posh$ and then “project” the results down to $\stsph$.
%%%%%
Accordingly, we need a suitable extension of the monotone quantum metric tensors to $\posh$,  very much in the spirit of Campbell's work on the extension of the Fisher-Rao metric tensor to the non-normalized case of finite measures.
%%%%%
Kumagai already investigated this problem and provide a complete solution of Petz's problem when the normalization condition on quantum states is lifted \cite{Kumagai-2011}. 
%%%%%%
Quite interestingly, the result very much resembles Campbell's result in the sense that the difference with the normalized case is  entirely contained in a function $b\colon \mathbb{R}^{+}\ra\mathbb{R}$    and a family $\{f_{t}\}_{t\in\mathbb{R}^{+}}$ of operator monotone functions satisfying $t b(t) + \frac{1}{f_{t}(1)}>0$.
%%%%%

In our case, however, it is not necessary to exploit the full level of generality of Kumagai's work.
%%%%
It suffices to find a Riemannian metric tensor $\hat{G}_{f}$ on $\posh$ such that
\be\label{eqn:metric_cone_restricts_to_metric_stsp}
i^{*}\hat{G}_{f}= G_{f}
\ee
where $i\colon\stsph\ra\posh$ is the canonical immersion and $ G_{f}$  is a monotone quantum metric tensor as in equation \eqref{eqn:Petz_characterization}.
%%%%%%%%%%%%%%%%%%%%%%%%%%%%%%%%%
Accordingly, we consider $\hat{G}_{f}$ as given by
%%%%%%%%%%%%%%%%%%%%%%%%%%%%%%
\begin{equation}\label{eqn:definition_monotone_metric_cone}
	\left( \hat{G}_{f} \right)_\omega ( \mathbf{v}_{\omega}, \mathbf{w}_{\omega}) = \kappa \Tr \left(  \mathbf{v}_{\omega} \left( \mathbf{ K}^f_\omega \right)^{-1} ( \mathbf{w}_{\omega})    \right),
\end{equation} 
%%%%%%%%%%%%%%%%%%%%%%%%%%%%%%
where  $f$ is the operator monotone function appearing in equation \eqref{eqn:Petz_characterization}   (and thus satisfying equation \eqref{eqn: monotone function}), $\omega \in \pos$, $ \mathbf{v}_{\omega},\mathbf{w}_{\omega} \in T_\omega \mathcal{ P ( H)} \cong \mathcal{ B}_{sa} (\hh)$, and $\mathbf{ K}^f_\omega$ is as in equation \eqref{eqn:superoperator_K_expression}.
%%%%%
Equation \eqref{eqn:definition_monotone_metric_cone} corresponds to the choice $b=0$ and $f_{t}=f$ in Kumagai's classification.
%%%%%

If we introduce the operators $\mathbf{e}^{\omega}_{lm}$ diagonalizing $\omega$, we can proceed as in the normalized case to obtain an equations analogous to  equation \eqref{eqn: Petz superoperator} so that,  recalling equation \eqref{eqn: monotone function}, we immediately obtain
\be\label{eqn: K superoperator on Delta}
\mathbf{ K}^f_{\omega}\left(\Delta \left( \omega \right)\right)=\mathbf{ K}^f_{\omega}\left(\omega\right)=\omega^{2} \Lra \left(\mathbf{ K}^f_{\omega}\right)^{-1}\left(\Delta \left(\omega \right)\right)=\mathbb{I} .
\ee
%%%%%%

%%%%%%%%%%%%%%%%%%%%%%%%%%%%%%
%%%%%%%%%%%%%%%%%%%%%%%%%%%%%%
%%%%%%%%%%%%%%%%%%%%%%%%%%%%%%
%%%%%%%%%%%%%%%%%%%%%%%%%%%%%%
%%%%%%%%%%%%%%%%%%%%%%%%%%%%%%
%%%%%%%%%%%%%%%%%%%%%%%%%%%%%%

\section{Lie groups and monotone quantum metric tensors}\label{sec: main results}

We are interested in classifying all those  actions of  $\Glh$ and $T^* \Uh$ on $\stsph$ that behaves in the way described in the introduction  with respect to suitable monotone quantum metric tensors.
%%%%
Specifically, we want to find all those actions, say $\delta$,  of either $\Glh$ or $T^* \Uh$ on $\stsph$ for which there is a monotone metric tensor $G_{f}$ on $\stsph$ such that the fundamental vector fields $X_{b}$ of the standard action $\alpha$ of $\Uh$ on $\stsph$ together with the gradient vector fields $Y_{a}^{f}$ associated with the expectation value functions $l_{a}$ close on a representation of the Lie algebra of either  $\Glh$ or $T^* \Uh$ that integrates to the action $\delta$.
%%%%%
From the results in \cite{Ciaglia-2020} we know that there are at least 3 monotone metric tensors for which this construction is possible for any finite-level quantum system.
%%%%
Moreover, from the results in \cite{C-DN-2021} we know that in the case of a two-level quantum systems, the Lie groups $\Glh$ and $T^* \Uh$ are the only Lie groups for which the construction described above is actually possible.
%%%%%%
Here, we want to understand if the group actions of $\Glh$ and $T^* \Uh$ found in \cite{C-DN-2021} can be extended from a 2-level quantum system to a system with an arbitrary, albeit finite, number of levels.
%%%%%
 
At this purpose, it is important to recall all those properties, shared by $\Glh$ and $T^* \Uh$ and by their actions, that are at the heart of the results of \cite{Ciaglia-2020,C-DN-2021}.
%%%%%
First of all, both $\Glh$ and $T^* \Uh$ contain the Lie group $\Uh$ as a Lie subgroup, and contain the elements $\lambda\mathbb{I}$ with $\lambda>0$ and $\mathbb{I}$ the identity operator on $\hh$.
%%%%
Then, all the (transitive) actions of both $\Glh$ and $T^* \Uh$ on $\stsph$ appearing in the analysis of \cite{Ciaglia-2020,C-DN-2021} arise as a sort of normalization  of suitable (transitive) actions on $\mathcal{P}(\hh)$.
%%%%
Specifically, if $G$ denotes either $\Glh$ or $T^* \Uh$, then every $G$-action $\delta$ on $\stsph$ can be written as
\be\label{eqn: normalized action}
\delta(\gr,\rho)=\frac{\hat{\delta}(\gr, \rho)}{\Tr(\hat{\delta}(\gr,\rho))}
\ee
with $\hat{\delta}$ a $G$-action on $\posh$ satisfying
\be
\hat{\delta}(\gr,\lambda\omega)=\lambda\hat{\delta}(\gr,\omega)
\ee
for every $\gr\in G$, for every $\omega\in\posh$, and for every $\lambda>0$.
%%%%%
Moreover, among all those actions $\hat{\delta}$ satisfying the properties discussed above, there is a preferred action $\hat{\delta}_{0}$ (the action $\hat{\beta}$ in equation \eqref{eqn: linear action of GL(H)}  for $\Glh$, and the action $\hat{\gamma}$ in equation \eqref{eqn: BKM-action on plf} for $T^{*}\Uh$) such that every relevant action $\hat{\delta}$ can be written as
\be\label{eqn: deformed action}
\hat{\delta}_{\phi} = \phi^{-1}\circ\hat{\delta}_{0} \circ \left(\mathrm{Id}_{G}\times \phi\right)
\ee
with $\phi\colon\posh\ra\posh$ a smooth diffeomorphism arising from a smooth diffeomorphism $\phi\colon \mathbb{R}^{+}\ra\mathbb{R}^{+}$ by means of functional calculus and such that
\be\label{eqn: properties of phi}
\phi\left(\mathbf{U}\omega\mathbf{U}^{\dagger}\right)=\mathbf{U}\phi(\omega)\mathbf{U}^{\dagger}  
\ee
and 
\be\label{eqn: properties of phi 2}
\phi( \omega ) = \sum_{j=1}^n   \phi( \omega_j )   \ket{e_j} \bra{e_j} 
\ee
where $\{ \ket{e_1}, \ket{e_2},\dots,\ket{e_n}  \}$ is a basis of $\hh$ made of eigenvectors of $\omega$.
%%%%%

Equation \eqref{eqn: deformed action} implies that the map $\phi$ is equivariant with respect to the action $\hat{\delta}_{\phi}$ and $\hat{\delta}_{0}$, which in turn implies that the fundamental vector fields of $\hat{\delta}_{\phi}$ are $\phi$-related with that of $\hat{\delta}_{0}$ ({\itshape cfr.} chapter 5 in \cite{A-M-R-2012}).
%%%%
By the very definition of $\phi$-relatedness ({\itshape cfr.} chapter 4 in \cite{A-M-R-2012}), denoting with $\zeta^{\phi}$ a fundamental vector field of $\hat{\delta}_{\phi}$ and with $\zeta$ a fundamental vector field of $\hat{\delta}_{0}$, it follows that
\be\label{eqn: phi-related fundamental vector fields}
\zeta^{\phi} = T(\phi^{-1}) \circ \zeta  \circ \phi .
\ee
%%%%%

We exploit equation \eqref{eqn: phi-related fundamental vector fields} to explicitly describe how the fundamental vector fields $\hat{Y}_{\mathbf{a}}$ of the action $\hat{\beta}$ of  $\Glh$  on $\posh$  \cfr{equation \eqref{eqn: linear action of GL(H)} and equation \eqref{eqn:definition_X_and_Y_PH}} transform under $\phi$.
%%%%
We then equate the result with the gradient vector field associated with the expectation value function $\hat{l}_{\mathbf{a}}$ by means of the metric tensor $\hat{G}_{f}$ as in equation \eqref{eqn:definition_monotone_metric_cone} thus obtaining  an explicit characterization of  the diffeomorphism $\phi$ and the operator monotone function $f$ compatible with the equality.
%%%% 
Finally, with this choice of $\phi$ and $f$, we prove that the gradient vector fields $Y^{f}_{\mathbf{a}}$ associated with the expectation value functions $l_{\mathbf{a}}$ on $\stsph$ by means of the monotone quantum metric $G_{f}$ as in equation \eqref{eqn:Petz_characterization} correspond to the fundamental vector fields $Z^{\phi}_{\mathbf{a}\,\mathbf{0}}$ of the action $\beta_{\phi}$ of $\Glh$ on $\stsph$ associated with the action $\hat{\beta}_{\phi}$ on $\posh$.
%%%%%

A similar procedure is then applied to the fundamental vector fields  
 $\hat{W}_{\mathbf{a} \, \mathbf{0}}$ of the action $\hat{\gamma}$ of $T^{*}\Uh$ on $\posh$ \cfr{equation \eqref {eqn: BKM-action on plf} and equation \eqref{eqn: gradient vector fields of BKM on P(H)}}.
 %%%%%

\subsection{The general linear group}\label{subsec: Gl(H)}

Following \cite{Ciaglia-2020,C-DN-2021}, when considering the general linear group $\Glh$, the reference action $\hat{\delta}_{0}$ appearing in equation \eqref{eqn: deformed action} is the action  $\hat{\beta}$ in equation \eqref{eqn: linear action of GL(H)}.
%%%%
Therefore, denoting with $\hat{Z}_{ \mathbf{ a} \, \mathbf{ b}}$ a fundamental vector field of $\hat{\delta}_{0}$ and with $\hat{Z}_{ \mathbf{ a} \, \mathbf{ b}}^{\phi}$ a fundamental vector field of $\hat{\delta}_{\phi}$, from equation  \eqref{eqn:fundamental_vector_fields_beta}, equation \eqref{eqn: phi-related fundamental vector fields}, and \cite[Thm. 5.3.1]{Bhatia-2007}, it follows that 
\be
\hat{Z}_{ \mathbf{ a} \, \mathbf{ 0}}^{\phi}(\omega) =  (\phi^{-1})^{[1]}\left( \phi ( \omega) \right) \square \{ \mathbf{ a}, \phi ( \omega)\},
\ee 
where $\square$ denotes the Schur product with respect to the basis of eigenvectors of $\phi(\omega)$,  and  
\begin{equation}
		(\phi^{-1})^{[1]}\left( \phi ( \omega) \right) = \sum_{\omega_j = \omega_k} \frac{1}{ \phi'(\omega_{j})} \, \ket{e_j} \bra{e_k} +  \sum_{\omega_j \ne \omega_k} \frac{\omega_j - \omega_k}{ \phi(\omega_j)   -  \phi( \omega _k)} \ket{e_j}\bra{e_k} ,
\end{equation}
with  $ \phi(\omega _j)$ the eigenvalues of $\phi(\omega)$ and with $\{ \ket{e_1}, \ket{e_2}, \dots, \ket{e_n} \}$ the basis of $\hh$ of eigenvectors of $\phi(\omega)$ and $\omega$ \cfr{equation  \eqref{eqn: properties of phi 2}}.
%%%%%%%%%%%%%%%%%%%%%%%%%%%%%%%%%%
Moreover, a direct computation shows that
\begin{equation}
		\{ \phi ( \omega), \mathbf{ a}\} = \frac{1}{2} \sum_{j,k} \left( \phi(\omega_j) + \phi( \omega_k) \right) a_{jk}  \ket{e_j}\bra{e_k},
	\end{equation}
where $a_{jk}$ are the components of $\mathbf{a}$ in the basis given by the eigenvectors of $\omega$.
%%%%
We thus conclude that
\begin{equation}\label{eqn: gradient vector field Gl(H) 0}
\hat{Z}_{ \mathbf{ a} \, \mathbf{ 0}}^{\phi}( \omega) = \sum_{\omega_j = \omega_k}  a_{jk} \frac{\phi( \omega_j)}{\phi'( \omega_j)} \ket{e_j} \bra{e_k} + \frac{1}{2} \sum_{\omega_j \ne \omega_k} a_{jk} \left( \phi(\omega_j) + \phi( \omega_k) \right) \frac{\omega_j - \omega_k}{\phi(\omega_j) - \phi( \omega_k)} \ket{e_j}\bra{e_k}.
\end{equation}
%%%%%

Now, we require that $\hat{Z}_{ \mathbf{ a} \, \mathbf{ 0}}^{\phi}$ is  the gradient vector field of the expectation value function $\hat{l}_\mathbf{a}$ with respect to the metric tensor $\hat{G}_f$ defined as in Subsection \ref{subsec:geometry_of_the_space_of_quantum_states_information_metrics} in order to characterize the function $f$.
%%%%%%
From the very definition of gradient vector field, it follows that
%%%%%%%%%%%%%%%%%%%%%%%%%%%%%%
\begin{equation}\label{eqn: gradient vector field Gl(H) 1}
d l_\mathbf{a}(\Gamma)\big|_{\omega} = \left( G_f \right)_\omega ( \hat{Z}_{ \mathbf{ a} \, \mathbf{ 0}}^{\phi}( \omega), \Gamma( \omega)) = \kappa \Tr \left(  \Gamma(\omega) \left( \mathbf{ K}^f_\omega \right)^{-1} \left( \hat{Z}_{ \mathbf{ a} \, \mathbf{ 0}}^{\phi}( \omega) \right) \right)
	\end{equation} 
holds for any vector field $\Gamma$    on $\mathcal{ P ( H)}$.
%%%%%%%%%%%%%%%%%%%%%%%%%%%%%%
On the other hand, it also holds that
\begin{equation}\label{eqn: gradient vector field Gl(H) 2}
		d l_\mathbf{a} (\Gamma)\big|_{\omega} =  \Gamma( l_\mathbf{a}) |_{\omega} = \Tr{ \left( \mathbf{ a} \Gamma ( \omega) \right)},
\end{equation}
 so that, comparing equation \eqref{eqn: gradient vector field Gl(H) 1} with equation \eqref{eqn: gradient vector field Gl(H) 2}, we obtain
\begin{equation} \label{eqn: gradient vector field Gl(H) 3}
	\hat{Z}_{ \mathbf{ a} \, \mathbf{ 0}}^{\phi}( \omega) =  \kappa^{-1} \mathbf{K}^f_\omega (\mathbf{ a}) = \kappa^{-1} f( \mathbf{ L}_\omega \mathbf{ R}^{-1}_\omega) \mathbf{ R}_\omega ( \mathbf{ a}).
	\end{equation}
%%%%%%%%%%%%%%%%%%%%%%%%%%%%%%

Exploiting equation \eqref{eqn: Petz superoperator} it follows that equation \eqref{eqn: gradient vector field Gl(H) 3} becomes
\be\label{eqn: gradient vector field Gl(H) 4}
\begin{split}
\hat{Z}_{ \mathbf{ a} \, \mathbf{ 0}}^{\phi}( \omega) & =  \sum_{\omega_{j} = \omega_{k}}   \frac{\omega_j a_{jk}}{\kappa} \ket{e_j} \bra{e_k} + \sum_{\omega_{j}\neq\omega_{k}} \frac{\omega_{k}}{\kappa }   f \left( \frac{\omega_j}{\omega_k}\right) \ket{e_j} \bra{e_k} .
\end{split}
\ee
%%%%%%%%%%%%%%%%%%%%%%%%%%%%%%
Comparing equation \eqref{eqn: gradient vector field Gl(H) 0} with equation \eqref{eqn: gradient vector field Gl(H) 4},  we obtain  
\be\label{eqn: phi function GL(H)}
\frac{\phi(\omega_j)}{\phi'(\omega_j)} = \kappa^{-1} \omega_j  \quad \ee
and 
\be\label{eqn: f function GL(H)}
\left( \phi(\omega_j) + \phi( \omega_k) \right) \frac{\omega_j - \omega_k}{\phi(\omega_j) - \phi( \omega_k)}= \frac{\omega_{k}}{\kappa }   f \left( \frac{\omega_j}{\omega_k}\right).
\ee
%%%%
Equation \eqref{eqn: phi function GL(H)} implies
\be\label{eqn: phi function GL(H) 2}
 \phi(x) = c\,  x^{\kappa},
\ee
with $c>0$, so that, because of equation \eqref{eqn: f function GL(H)},  the function $f$ in $\hat{G}_{f}$ must be of the form
\begin{equation} \label{eqn: function f 2}
f ( x) = \frac{\kappa}{2} \frac{ (x - 1) (x^\kappa + 1)}{ x^\kappa - 1}.
\end{equation}
%%%%%%%%%%%%%%%%%%%%%%%%%%%%%%
A direct check shows that the function $f$ in equation \eqref{eqn: function f 2} satisfies the properties listed in equation \eqref{eqn: monotone function} for all $\kappa>0$, but we do not know if it is operator monotone for every $\kappa>0$.
%%%%%
The following proposition shows that $f$ is operator monotone if and only if $0<\kappa\leq 1$.
%%%%%

\begin{proposition}\label{prop: opmon}
The function $f$ in equation \eqref{eqn: function f 2} is operator monotone if and only if $0<\kappa<1$.
\end{proposition}
\begin{proof}

When $\kappa=1$ it is $f(x)=\frac{1+x}{2}$ which is known to be operator monotone and to be associated with the Bures-Helstrom metric tensor \cite{Petz-1996}.
%%%%%

%Since $f(0)=\frac{\kappa}{2}$ and $f(1)=1$, it follows that $f$ can not be operator monotone whenever $\kappa\geq 2$ because then $f(0)\geq f(1)$ with $f$ a non-constant function.
%%%%%

The function $f$ as in equation \eqref{eqn: function f 2} is clearly $C^{1}$ in $(0,+\infty)$ and it is continuous in $[0,+\infty)$.
%%%%
When $ \kappa>1$, it holds
\be\label{eqn: opmon}
\lim_{x\ra 0^{+}}f'(x)=\lim_{x\ra 0^{+}} \frac{\kappa}{2}\,\left(-\frac{1+x^{\kappa}}{1-x^{\kappa}} + 2\kappa\frac{x^{\kappa -1}  (1-x)}{(x^{\kappa}-1)^{2}}\right)= -\frac{\kappa}{2}<0
\ee
which means that there is $\epsilon> 0$ such that $f(x)$ is decreasing for $x\in(0,\epsilon)$ and thus $f$ can not be operator monotone.
%%%%
Note that \eqref{eqn: opmon} is no longer valid when $0<\kappa < 1$ because of the term $	x^{\kappa -1}$

%%%%

When $\kappa\in (0,1)$ we consider the rational case $\kappa=\frac{k}{n}$ with $k < n$ since the passage to an irrational $\kappa\in (0,1)$ is obtained by continuity just as in \cite[Prop. 3.1]{Furuta-2008}.
%%%%
Following \cite[Prop. 3.1]{Furuta-2008}, we write
\be\label{eqn: opmon1}
1-x^{\frac{k}{n}}=(1-x^{\frac{1}{n}})\,\sum_{l=0}^{k-1}x^{\frac{l}{n}}
\ee
so that
\be\label{eqn: opmon2}
\begin{split}
f(x)&=\frac{k}{2n} \frac{ (x - 1) (x^{\frac{k}{n}} + 1)}{ x^{\frac{k}{n}} - 1}=\frac{k(x^{\frac{k}{n}} + 1)}{2n}\left(\frac{\sum_{l=0}^{n-1}x^{\frac{l}{n}}}{\sum_{j=0}^{k-1}x^{\frac{j}{n}}}\right)=\frac{k(x^{\frac{k}{n}} + 1)}{2n}\left(1+\frac{\sum_{l=k}^{n-1}x^{\frac{l}{n}}}{\sum_{j=0}^{k-1}x^{\frac{j}{n}}}\right)= \\
&= \frac{k}{2n}\left(x^{\frac{k}{n}} + 1 + \sum_{l=k}^{n-1}\left(\sum_{j=0}^{k-1} x^{\frac{j-l}{n}}\right)^{-1} + \sum_{l=k}^{n-1}\left(\sum_{j=0}^{k-1}x^{\frac{j-l-k}{n}}\right)^{-1}\right) .
\end{split}
\ee
%%%%
Since $j<k<l<n$, the functions
\be
g(x)=\left(\sum_{j=0}^{k-1} x^{\frac{j-l}{n}}\right)^{-1} \quad \mbox{and}\quad h(x)=\left(\sum_{j=0}^{k-1}x^{\frac{j-l-k}{n}}\right)^{-1}
\ee
are operator monotone according to \cite[Thm. LH-1]{Furuta-2008}, and thus the function $f$ in equation \eqref{eqn: opmon2} is operator monotone because it is the sum of operator monotone functions.
%%%%%

%%%%%

\end{proof}

Finally, when $\phi$ is as in equation \eqref{eqn: phi function GL(H)} and  $f$ is as in equation \eqref{eqn: function f 2}, we prove that the fundamental vector fields $Z_{\mathbf{a}\,\mathbf{0}}^{\phi}$ of the normalized action $\beta_{\phi}$ of $\Glh$ on $\stsph$ associated with $\hat{\beta}_{\phi}$ by means of equation \eqref{eqn: normalized action} are indeed the gradient vector fields associated with the expectation value functions $l_{\mathbf{a}}$ by means of the monotone metric tensor $G_{f}$.
%%%%%
Indeed, from equation  \eqref{eqn: normalized action} it follows that \be\label{eqn:definition_W_a}
%%%%%%%%%%%%%%%%%%%%%%%%%%%%%%
\begin{split}
Z_{\mathbf{a}\,\mathbf{0}}^{\phi}( \rho) & = \frac{d}{dt}\beta_{\phi} \left( \exp\left( \frac{t}{2} ( \mathbf{a}, \mathbf{0}) \right), \rho \right)\Big|_{t = 0}   = \frac{d}{dt}  \frac{ \hat{\beta}_{\phi} \left( \exp\left( \frac{t}{2} ( \mathbf{a}, \mathbf{0}) \right), i ( \rho) \right)}{ \Tr \left(\hat{\beta}_{\phi} ( \exp\left( \frac{t}{2} ( \mathbf{a}, \mathbf{0}) \right) , i( \rho)) \right)}\Bigg|_{t = 0} = \\
& = \hat{Z}_{ \mathbf{ a} \, \mathbf{ 0}}^{\phi}(i(\rho)) -  \Tr\left(\hat{Z}_{ \mathbf{ a} \, \mathbf{ 0}}^{\phi}(i(\rho)) \right)\,\Delta(i(\rho)) .
\end{split}
\end{equation}
%%%%%
Equation \eqref{eqn:definition_W_a} is equivalent to
\begin{equation} \label{eqn:relation_fundamental_vector_fields_alpha_S_and_alpha_P_geom}
T_{\rho}i\left(Z_{\mathbf{a}\,\mathbf{0}}^{\phi}( \rho)\right) = \hat{Z}_{ \mathbf{ a} \, \mathbf{ 0}}^{\phi}(i(\rho)) -  \Tr\left(\hat{Z}_{ \mathbf{ a} \, \mathbf{ 0}}^{\phi}(i(\rho)) \right)\,\Delta(i(\rho))  
\end{equation}
for all $\mathbf{a}\in \bhsa$ and all $\rho\in\stsph$, and this last instance is  equivalent to the fact that $Z_{\mathbf{a}\,\mathbf{0}}^{\phi}$ is $i$-related with the vector field $\hat{Z}_{ \mathbf{ a} \, \mathbf{ 0}}^{\phi} - (\mathcal{L}_{\hat{Z}_{ \mathbf{ a} \, \mathbf{ 0}}^{\phi}}\hat{l}_{\mathbb{I}})\Delta$ for all $\mathbf{a}\in \bhsa$.
%%%

We now set 
\be
\hat{Y}^{\phi}_{\mathbf{a}}:=\hat{Z}_{ \mathbf{ a} \, \mathbf{ 0}}^{\phi}, \quad \mbox{ and } \quad  Y^{\phi}_{\mathbf{a}}:=Z_{ \mathbf{ a} \, \mathbf{ 0}}^{\phi} .
\ee
To finish the proof of the proposition, we need to prove that $Y^{\phi}_{\mathbf{a}}$ is actually the gradient vector field of the expectation value function $l_{\mathbf{a}}$ for every $\mathbf{a}\in\bhsa$.
%%%%
At this purpose, we compute
\be\label{eqn: relation between gradient vector fields of hat(delta) and delta}
\begin{split}
(G_{f}(Y^{\phi}_{\mathbf{a}},V))(\rho)&=(i^{*}\hat{G}_{f}(Y^{\phi}_{\mathbf{a}},V))(\rho)=(\hat{G}_{f})_{i(\rho)}(T_{\rho}i(Y_{\mathbf{a}}(\rho)),T_{\rho}i(V(\rho)))= \\
&\stackrel{\mbox{\eqref{eqn:relation_fundamental_vector_fields_alpha_S_and_alpha_P_geom}}}{=} (\hat{G}_{f})_{i(\rho)}\left(\hat{Y}^{\phi}_{\mathbf{a}} \left( i(\rho) \right)  - \Tr \left( \hat{Y}^{\phi}_{\mathbf{a}}\left( i (\rho) \right) \right)  \Delta \left( i(\rho) \right),T_{\rho}i(V(\rho))\right).
\end{split}
\ee
%%%%
Since we proved that $\hat{Y}^{\phi}_{\mathbf{a}}$ is be the gradient vector field associated with $\hat{l}_{\mathbf{a}}$ by means of $\hat{G}_{f}$, equation \eqref{eqn: relation between gradient vector fields of hat(delta) and delta} becomes
\be\label{eqn: relation between gradient vector fields of hat(delta) and delta 2}
\begin{split}
(G_{f}(Y^{\phi}_{\mathbf{a}},V))(\rho)&=\left(\mathcal{L}_{V} l_{\mathbf{a}}\right)(\rho)  - \Tr \left( \hat{Y}^{\phi}_{\mathbf{a}}\left( i (\rho) \right) \right)  (\hat{G}_{f})_{i(\rho)}\left(\Delta \left( i(\rho) \right),T_{\rho}i(V(\rho))\right).
\end{split}
\ee
%%%%%%
The second term in the right-hand-side of equation \eqref{eqn: relation between gradient vector fields of hat(delta) and delta 2} vanishes.
%%%%%
Indeed, equation  \eqref{eqn:definition_monotone_metric_cone} implies that
\be\label{eqn: relation between gradient vector fields of hat(delta) and delta 3}
(\hat{G}_{f})_{i(\rho)}\left(\Delta \left( i(\rho) \right),T_{\rho}i(V(\rho))\right)=  \Tr \left(T_{\rho}i(V(\rho))  \left( \mathbf{ K}^f_{i (\rho)} \right)^{-1}  \left(\Delta \left( i(\rho) \right) \right)    \right).
\ee
%%%%%
From equation \eqref{eqn: K superoperator on Delta} we conclude that equation \eqref{eqn: relation between gradient vector fields of hat(delta) and delta 3} becomes
\be\label{eqn: relation between gradient vector fields of hat(delta) and delta 4}
(\hat{G}_{f})_{i(\rho)}\left(\Delta \left( i(\rho) \right),T_{\rho}i(V(\rho))\right)=  \Tr \left(T_{\rho}i(V(\rho)) \right)= 0,
\ee
where the last equality follows from equation \eqref{eqn: tangent space at a state}.
%%%%%
Inserting equation \eqref{eqn: relation between gradient vector fields of hat(delta) and delta 4} in equation \eqref{eqn: relation between gradient vector fields of hat(delta) and delta 2} we obtain
\be\label{eqn: relation between gradient vector fields of hat(delta) and delta 5}
\begin{split}
(G_{f}(Y^{\phi}_{\mathbf{a}},V))(\rho)&=\left(\mathcal{L}_{V} l_{\mathbf{a}}\right)(\rho)  
\end{split}
\ee
for every fundamental vector field of $\beta_{\phi}$ of the type $Y^{\phi}_{\mathbf{a}}$, for every vector field $V$ on $\stsph$ and for and every $\rho\in\stsph$.
%%%%
Equation \eqref{eqn: relation between gradient vector fields of hat(delta) and delta 5} is equivalent to the fact that $Y^{\phi}_{\mathbf{a}}$ is the gradient vector field associated with the expectation value function $l_{\mathbf{a}}$ by means of $G_{f}$ for every $\mathbf{a}\in\bhsa$ as desired. 
%%%%%%

Collecting the results proved in this subsection we obtain the following proposition.
%%%%

\begin{proposition}\label{prop: GL(H) actions and monotone quantum metric tensors}
The function $f$ given by
\begin{equation} \label{eqn: function f 3}
f ( x) = \frac{\kappa}{2} \frac{ (x - 1) (x^\kappa + 1)}{ x^\kappa - 1} 
\end{equation}
is operator monotone and satisfies equation \eqref{eqn: monotone function} if and only if   $0 < \kappa\leq 1$.
%%%%%
In these cases, denoting with $\{X_{\mathbf{b}}\}_{\mathbf{b}\in\mathcal{B}_{sa}(\hh)}$  the fundamental vector fields of the canonical action $\alpha$ of $\Uh$ on $\stsph$ as in equation \eqref{eqn:definition_alpha_stsp},  if $G_{f}$ is the associated monotone quantum metric tensor on $\stsph$ as in equation \eqref{eqn:Petz_characterization} and   $Y^{f}_{\mathbf{a}}$ is the gradient vector field associated with the expectation value function $l_{\mathbf{a}}$ with $\mathbf{a}\in\mathcal{B}_{sa}(\hh)$, the family $\{Y^{f}_{\mathbf{a}},X_{\mathbf{b}}\}_{\mathbf{a},\mathbf{b}\in\mathcal{B}_{sa} (\hh)}$ of vector fields on $\stsph$  close an anti-representation of the Lie algebra of the general linear group $\Glh$  integrating to the group action 
\begin{equation} \label{eqn:family of actions of special linear group on states 2}
\beta^{\kappa}(\gr, \rho) = \frac{ (\gr \rho^{\sqrt{\kappa}} \gr^\dagger)^{\frac{1}{\sqrt{\kappa}}}} {\Tr\left( \left( \gr \rho^{\sqrt{\kappa}} \gr^\dagger \right)^{ \frac{1}{\sqrt{\kappa}}} \right)}.
\end{equation}
%%%%%
The action $\beta^{\kappa}$ in equation \eqref{eqn:family of actions of special linear group on states 2} is transitive on $\stsph$ for every $0<\kappa \leq 1$.
%%%%
In particular, when $\kappa=1$ we recover the Bures-Helstrom metric tensor and the action $\beta$ in equation \eqref{eqn:action of general linear group on states BH}, while when $\kappa=\frac{1}{4}$ we recover the Wigner-Yanase metric tensor and the action $\beta_{WY}$ in equation \eqref{eqn:action of general linear group on states WY}.
%%%%%
\end{proposition}

\subsection{The cotangent group of the unitary group}

Following what is done in subsection \ref{subsec: Gl(H)}, we consider an action $\hat{\gamma}_{\phi}$ associated with the action $\hat{\gamma}$ \cfr{equation \eqref{eqn: BKM-action on plf}} by means of equation \eqref{eqn: normalized action} with $\hat{\delta}_{0}\equiv\hat{\gamma}$.
%%%%
The fundamental vector fields $\hat{W}_{\mathbf{a}\,\mathbf{0}}^{\phi}$ of $\hat{\gamma}_{\phi}$ are   obtained as follows.
%%%%
From equation  \ \eqref{eqn: gradient vector fields of BKM on P(H)}, equation \eqref{eqn: phi-related fundamental vector fields}, and \cite[Thm. 5.3.1]{Bhatia-2007}, it follows that 
\be\label{eqn: T*U(H) vector fields}
\hat{W}_{ \mathbf{ a} \, \mathbf{ 0}}^{\phi}(\omega) =  (\phi^{-1})^{[1]}\left( \phi ( \omega) \right) \square \,\hat{W}_{\mathbf{a} \, \mathbf{0}}(\phi(\omega)),
\ee 
where $\square$ denotes the Schur product with respect to the basis of eigenvectors of $\phi(\omega)$,  and  
\be\label{eqn: T*U(H) vector fields 2} 
		(\phi^{-1})^{[1]}\left( \phi ( \omega) \right) = \sum_{\omega_j = \omega_k} \frac{1}{ \phi'(\omega_{j})}  \ket{e_j} \bra{e_k} +  \sum_{\omega_j \ne \omega_k} \frac{\omega_j - \omega_k}{ \phi(\omega_j)   -  \phi( \omega _k)} \ket{e_j}\bra{e_k} ,
\ee
with  $ \phi(\omega _j)$ the eigenvalues of $\phi(\omega)$ and with $\{ \ket{e_1}, \ket{e_2}, \dots, \ket{e_n} \}$ the basis of $\hh$ of eigenvectors of $\phi(\omega)$ and $\omega$ \cfr{equation  \eqref{eqn: properties of phi 2}}.
%%%%
On the other hand, from equation  \eqref{eqn: gradient vector fields of BKM on P(H)} and equation  \eqref{eqn: properties of phi 2} it follows that
\be\label{eqn: T*U(H) vector fields 3}
\begin{split}
\hat{W}_{\mathbf{a} \, \mathbf{0}}(\phi(\omega))&= \int_{0}^{1}\,\mathrm{d}\lambda\,\left((\phi(\omega))^{\lambda}\,\mathbf{a} \,(\phi(\omega))^{1-\lambda}\right)=\\
&=\sum_{\omega_{j}= \omega_{k}} \phi(\omega_j) a_{jk} \ket{e_j}\bra{e_k} + \sum_{\omega_{j}\neq\omega_{k}}  \frac{ \phi(\omega_{j}) -\phi(\omega_{k} )}{\ln\left(\frac{\phi(\omega_{j})}{\phi(\omega_{k})}\right)} a_{jk}  \ket{e_j}\bra{e_k},
\end{split}
\ee
so that, exploiting equation \eqref{eqn: T*U(H) vector fields 2} and equation  \eqref{eqn: T*U(H) vector fields 3}, equation \eqref{eqn: T*U(H) vector fields} becomes
\be\label{eqn: T*U(H) vector fields 5}
\hat{W}_{ \mathbf{ a} \, \mathbf{ 0}}^{\phi}(\omega)=\sum_{\omega_j = \omega_k} \frac{\phi(\omega_j) }{ \phi'(\omega_{j})} a_{jk} \ket{e_j} \bra{e_k} +  \sum_{\omega_j \ne \omega_k} \frac{\omega_j - \omega_k}{\ln\left(\frac{\phi(\omega_{j})}{\phi(\omega_{k})}\right)} a_{jk} \ket{e_j}\bra{e_k} 
\ee
%%%%%

In analogy with what is done in subsection \ref{subsec: Gl(H)}, we now require that $\hat{W}_{ \mathbf{ a} \, \mathbf{ 0}}^{\phi}$ is  the gradient vector field of the expectation value function $\hat{l}_\mathbf{a}$ with respect to a metric tensor $\hat{G}_f$ defined as in Subsection \ref{subsec:geometry_of_the_space_of_quantum_states_information_metrics} in order to characterize the function $f$.
%%%%%%
From the very definition of gradient vector field, it follows that
%%%%%%%%%%%%%%%%%%%%%%%%%%%%%%
\begin{equation}\label{eqn: gradient vector field T*U(H) 1}
d l_\mathbf{a}(\Gamma)\big|_{\omega} = \left( G_f \right)_\omega ( \hat{W}_{ \mathbf{ a} \, \mathbf{ 0}}^{\phi}( \omega), \Gamma( \omega)) = \kappa \Tr \left(  \Gamma(\omega) \left( \mathbf{ K}^f_\omega \right)^{-1} \left( \hat{W}_{ \mathbf{ a} \, \mathbf{ 0}}^{\phi}( \omega) \right) \right)
	\end{equation} 
holds for any vector field $\Gamma$    on $\mathcal{ P ( H)}$.
%%%%%%%%%%%%%%%%%%%%%%%%%%%%%%
On the other hand, it also holds that
\begin{equation}\label{eqn: gradient vector field T*U(H) 2}
		d l_\mathbf{a} (\Gamma)\big|_{\omega} =  \Gamma( l_\mathbf{a}) |_{\omega} = \Tr{ \left( \mathbf{ a} \Gamma ( \omega) \right)},
\end{equation}
 so that, comparing equation \eqref{eqn: gradient vector field T*U(H) 1} with equation \eqref{eqn: gradient vector field T*U(H) 2}, we obtain
\begin{equation} \label{eqn: gradient vector field T*U(H) 3}
	\hat{W}_{ \mathbf{ a} \, \mathbf{ 0}}^{\phi}( \omega) =  \kappa^{-1} \mathbf{K}^f_\omega (\mathbf{ a}) = \kappa^{-1} f( \mathbf{ L}_\omega \mathbf{ R}^{-1}_\omega) \mathbf{ R}_\omega ( \mathbf{ a}).
	\end{equation}
%%%%%%%%%%%%%%%%%%%%%%%%%%%%%%

Exploiting equation \eqref{eqn: Petz superoperator} it follows that equation \eqref{eqn: gradient vector field T*U(H) 3} becomes
\be\label{eqn: gradient vector field T*U(H) 4}
\begin{split}
\hat{W}_{ \mathbf{ a} \, \mathbf{ 0}}^{\phi}( \omega) & =  \sum_{\omega_{j} = \omega_{k}}   \frac{\omega_j a_{jk}}{\kappa} \ket{e_j} \bra{e_k} + \sum_{\omega_{j}\neq\omega_{k}} \frac{\omega_{k}}{\kappa }   f \left( \frac{\omega_j}{\omega_k}\right) \ket{e_j} \bra{e_k} .
\end{split}
\ee
%%%%%%%%%%%%%%%%%%%%%%%%%%%%%%
Comparing equation \eqref{eqn: T*U(H) vector fields 5} with equation \eqref{eqn: gradient vector field T*U(H) 4},  we obtain  
\be\label{eqn: phi function T*U(H)}
\frac{\phi(\omega_j)}{\phi'(\omega_j)} = \kappa^{-1} \omega_j  \quad \ee
and 
\be\label{eqn: f function T*U(H)}
\frac{\omega_j - \omega_k}{\ln\left(\frac{\phi(\omega_{j})}{\phi(\omega_{k})}\right)}= \frac{\omega_{k}}{\kappa }   f \left( \frac{\omega_j}{\omega_k}\right).
\ee
%%%% 
Equation \eqref{eqn: phi function T*U(H)} implies
\be\label{eqn: phi function T*U(H) 2}
 \phi(x) = c\,  x^{\kappa},
\ee
with $c>0$, and it is worth noting that the family of diffeomorphisms found here is the same as that found in subsection \ref{subsec: Gl(H)} in the case of the general linear group $\Glh$ \cfr{equation \eqref{eqn: phi function GL(H) 2}}.
%%%%
Because of equation \eqref{eqn: phi function T*U(H) 2} and  equation \eqref{eqn: f function T*U(H)}, the function $f$ in $\hat{G}_{f}$ must be of the form
\be\label{eqn: function f 2 T*U(H)}
f ( x) = \kappa \,\frac{x -1 }{ \ln\left(x\right)},
\ee
which is precisely the operator monotone function associated with the Bogoliubov-Kubo-Mori metric tensor up to the constant $\kappa$ \cite{Petz-1996}.
%%%
Note that the positive constant $\kappa$ is here arbitrary differently from what happens for $\Glh$ \cfr{subsection \ref{subsec: Gl(H)}}.
%%%%%%

It is a matter of direct computation to check that the form of $\phi$ in equation \eqref{eqn: phi function T*U(H) 2} implies that the action $\hat{\gamma}_{\phi}$ associated with the action $\hat{\gamma}$ \cfr{equation \eqref{eqn: BKM-action on plf}} by means of equation \eqref{eqn: normalized action} with $\hat{\delta}_{0}\equiv\hat{\gamma}$ reads
\be
\hat{\gamma}_{\phi}((\mathbf{U},\mathbf{a}),\omega)=\hat{\gamma}\left(\left(\mathbf{U},\frac{\mathbf{a}}{\kappa}\right),\omega\right),
\ee
so that
\be
\hat{W}_{\mathbf{a} \, \mathbf{0}}^{\phi} = \hat{W}_{\frac{\mathbf{a}}{\kappa} \, \mathbf{0}}=\frac{1}{\kappa}\,\hat{W}_{ \mathbf{a} \, \mathbf{0}}
\ee
\cfr{equation \eqref{eqn: T*U(H) vector fields 3}, equation \eqref{eqn: T*U(H) vector fields 5}, and equation \eqref{eqn: phi function T*U(H) 2}}.
%%%%
Consequently, the fundamental vector fields $W_{\mathbf{a},\mathbf{0}}^{\phi}$ of the normalized action $\gamma_{\phi}$ associated with $\hat{\gamma}_{\phi}$ by means of equation \eqref{eqn: normalized action} read
\be\label{eqn:definition_W_a T*U(H)}
%%%%%%%%%%%%%%%%%%%%%%%%%%%%%%
\begin{split}
W_{\mathbf{a} \, \mathbf{0}}^{\phi}( \rho) & =       \frac{d}{dt}  \frac{ \hat{\gamma}_{\phi} \left( \exp\left( \frac{t}{2} ( \mathbf{a}, \mathbf{0}) \right), i ( \rho) \right)}{ \Tr \left(\hat{\gamma}_{\phi} ( \exp\left( \frac{t}{2} ( \mathbf{a}, \mathbf{0}) \right) , i( \rho)) \right)}\Bigg|_{t = 0} = \\
& = \hat{W}_{ \mathbf{ a} \, \mathbf{ 0}}^{\phi}(i(\rho)) -  \Tr\left(\hat{W}_{ \mathbf{ a} \, \mathbf{ 0}}^{\phi}(i(\rho)) \right)\,\Delta(i(\rho)) = \\
& = \frac{1}{\kappa}\left(\hat{W}_{ \mathbf{ a} \, \mathbf{ 0}}(i(\rho)) -  \Tr\left(\hat{W}_{ \mathbf{ a} \, \mathbf{ 0}} (i(\rho)) \right)\,\Delta(i(\rho))\right).
\end{split}
\end{equation}
%%%%%
Equation \eqref{eqn:definition_W_a T*U(H)} is equivalent to
\begin{equation}
T_{\rho}i\left(W_{\mathbf{a} \, \mathbf{0}}^{\phi}( \rho)\right) =  \frac{1}{\kappa}\left(\hat{W}_{ \mathbf{ a} \, \mathbf{ 0}} (i(\rho)) -  \Tr\left(\hat{W}_{ \mathbf{ a} \, \mathbf{ 0}} (i(\rho)) \right)\,\Delta(i(\rho)) \right) 
\end{equation}
for all $\mathbf{a}\in \bhsa$ and all $\rho\in\stsph$, and this last instance is  equivalent to the fact that $W_{\mathbf{a} \, \mathbf{0}}^{\phi}$ is $i$-related with the vector field $ \frac{1}{\kappa}\left(\hat{W}_{ \mathbf{ a} \, \mathbf{ 0}}  - (\mathcal{L}_{\hat{W}_{ \mathbf{ a} \, \mathbf{ 0}}} \hat{l}_{\mathbb{I}})\Delta\right)$ for all $\mathbf{a}\in \bhsa$.
%%%

Now, proceeding in complete analogy with what is done in subsection \ref{subsec: Gl(H)}, it is possible to prove that, when $\phi$ and $f$ are as in equation \eqref{eqn: phi function T*U(H) 2} and equation \eqref{eqn: function f 2 T*U(H)} respectively, then the fundamental vector field  $W_{\mathbf{a} \, \mathbf{0}}^{\phi}$ is the gradient vector field associated with the expectation value function $l_{\mathbf{a}}$ by means of the monotone quantum metric tensor $G_{f}$ (coinciding with the Bogoliubov-Kubo-Mori metric tensor up to the constant $\kappa$) for all $\mathbf{a}\in\bhsa$.
%%%%%
Collecting the results in this subsection we obtain the following proposition.
%%%%

\begin{proposition}\label{prop: T*U(H) actions and monotone quantum metric tensors}
Given the operator monotone function 
\be
f(x)=\kappa \,\frac{x -1 }{ \ln\left(x\right)}
\ee
satisfying equation \eqref{eqn: monotone function} and associated with the Bogoliubov-Kubo-Mori metric tensor $G_{f}\equiv \kappa G_{BKM}$ (up to the constant factor $\kappa>0$)  through equation \eqref{eqn:Petz_characterization} \cite{Petz-1996},  denoting with $\{X_{\mathbf{b}}\}_{\mathbf{b}\in\mathcal{B}_{sa}(\hh)}$  the fundamental vector fields of the canonical action $\alpha$ of $\Uh$ on $\stsph$ as in equation \eqref{eqn:definition_alpha_stsp},    and denoting with   $V^{f}_{\mathbf{a}}$  the gradient vector field associated with the expectation value function $l_{\mathbf{a}}$ with $\mathbf{a}\in\mathcal{B}_{sa} (\hh)$ by means of $G_{f}$, the family $\{V^{f}_{\mathbf{a}},X_{\mathbf{b}}\}_{\mathbf{a},\mathbf{b}\in\mathcal{B}_{sa} (\hh)}$ of vector fields on $\stsph$  close an anti-representation of the Lie algebra of the cotangent group $T^{*}\Uh$  integrating to the group action 
\begin{equation} \label{eqn:family of actions of T*U(H) on states 2}
\gamma^{\kappa}\left((\mathbf{U},\mathbf{a}),\rho\right):=\frac{ \mathrm{e}^{U \ln(\rho) U^\dagger + \frac{\mathbf{a}}{\kappa}}}{\Tr( \mathrm{e}^{U\ln(\rho) U^\dagger + \frac{\mathbf{a}}{\kappa}})}
\end{equation}
%%%%%
The action $\gamma^{\kappa}$ in equation \eqref{eqn:family of actions of T*U(H) on states 2} is transitive on $\stsph$ for every $ \kappa>0$.
%%%%
 
\end{proposition}

%%%%%%%%%%%%%%%%%%%%%%%%%%%%%%
%%%%%%%%%%%%%%%%%%%%%%%%%%%%%%
%%%%%%%%%%%%%%%%%%%%%%%%%%%%%%
\section{Conclusions}\label{sec: conclusions} 
%%%%%%%%%%%%%%%%%%%%%%%%%%%%%%
%%%%%%%%%%%%%%%%%%%%%%%%%%%%%%

There are several ways in which the results presented here can be further developed in order to fully understand how the 2-dimensional picture discussed in \cite{C-DN-2021} extends to arbitrary finite dimensions.
%%%%%

First of all, concerning the Lie group $\Glh$, it is necessary to understand if there exist smooth transitive actions on $\posh$ that are not of the form $\beta_{\phi}$ \cfr{equation \eqref{eqn: linear action of GL(H)} and equation \eqref{eqn: deformed action}}.
%%%%
Then, it is necessary to understand if there exist smooth transitive actions on $\stsph$ that do not arise from smooth actions of $\Glh$ on $\posh$ as in equation \eqref{eqn: normalized action}.
%%%%%
If the answer to both these questions are negative, than it follows that the only actions of $\Glh$ on $\stsph$ whose associated Lie algebra anti-representations can be described in terms of the fundamental vector fields of the standard action of $\Uh$ on $\stsph$ \cfr{equation \eqref{eqn:definition_alpha_stsp} and equation \eqref{eqn:definition_X_and_Y_stsp}} and the gradient vector fields $Y^{f}_{\mathbf{a}}$ associated with the expectation value functions $l_{\mathbf{a}}$ by means of a suitable monotone quantum metric tensor are those found in this work.
%%%%%%%%%%%%%%%%%%%%%%%%%%%%%%%%%%

Concerning the group $T^{*}\Uh$, it is necessary to understand if there exist smooth transitive actions on $\stsph$ that do not arise from smooth actions of $T^{*}\Uh$ on $\posh$ as in equation \eqref{eqn: normalized action}.
%%%%%%%%%%%%%%%%%%%%%%%%
If the answer to this question is negative, than it follows that the only action  of $T^{*}\Uh$ on $\stsph$ whose associated Lie algebra anti-representations can be described in terms of the fundamental vector fields of the standard action of $\Uh$ on $\stsph$ \cfr{equation \eqref{eqn:definition_alpha_stsp} and equation \eqref{eqn:definition_X_and_Y_stsp}} and the gradient vector fields associated with the expectation value functions $l_{\mathbf{a}}$ by means of a suitable monotone quantum metric tensor are the ones found in this work, that is, the one associated with the Bogoliubov-Kubo-Mori metric tensor.
%%%%%%%%%%%%%%%%%%%%%%%%%%%%%%%%

%%%%%%%%%%%%%%%%%%%%%%%%%%%%%%%%%%%%%%%%
Besides the cases involving the Lie groups $\Glh$ and $T^{*}\Uh$, it is also necessary to understand if, for a quantum system whose Hilbert space $\hh$ has dimension greater than 2, there exists other Lie groups acting smoothly and transitively on $\stsph$ and whose Lie algebra anti-representation  can be described in terms of the fundamental vector fields of the standard action of $\Uh$ on $\stsph$ \cfr{equation \eqref{eqn:definition_alpha_stsp} and equation \eqref{eqn:definition_X_and_Y_stsp}} and the gradient vector fields  associated with the expectation value functions $l_{\mathbf{a}}$ by means of  suitable monotone quantum metric tensors.
%%%%%%
Concerning this instance, something can be said on some general properties any such Lie group $\gapp$ must possess.
%%%%%
First of all, the unitary group $\Uh$ must appear as a subgroup of $\gapp$ and $\mathrm{dim}(\gapp)=2\mathrm{dim}(\Uh)$.
%%%%
This last condition follows from the fact that the gradient vector fields associated with the expectation value functions $l_{\mathbf{a}}$ are labelled by elements in $\bhsa$, and thus the dimension of the Lie algebra $\mathfrak{g}$ of $\gapp$ is twice that of the Lie algebra of $\Uh$.
%%%%%
From this last observation it also follows that 
\be\label{eqn: Lie algebra of gapp}
\mathfrak{g}\cong\uh\oplus\bhsa
\ee 
as a vector space.
%%%%
Moreover, since $\Uh$ must be a subgroup of $\gapp$, there must be a decomposition of $\mathfrak{g}$ as in equation \eqref{eqn: Lie algebra of gapp} for which $\uh\oplus\{\mathbf{0}\}$ is a Lie subalgebra isomorphic to $\uh$.
%%%%%
Then, as already argued in \cite{Ciaglia-2020}, the requirement that the fundamental vector fields $X_{\mathbf{b}}$ of the standard action $\alpha$ of $\Uh$ on $\stsph$ are Killing vector fields for every monotone quantum metric tensor $G_{f}$ imposes additional constraints on the possible commutator between these vector fields and the gradient vector fields $Y^{f}_{\mathbf{a}}$ associated with the expectation value functions $l_{\mathbf{a}}$.
%%%%
Specifically, since $Y^{f}_{\mathbf{a}}$ is the gradient vector field associated with the expectation value function  $l_{\mathbf{a}}$ for every $\mathbf{a}\in\bhsa$, it follows that
\be\label{eqn: constraints mixed brackets}
\begin{split}
\mathcal{L}_{[X_{\mathbf{b}},Y^{f}_{\mathbf{a}}]}l_{\mathbf{c}}&=\mathcal{L}_{ X_{\mathbf{b}}}\left(\mathcal{L}_{Y^{f}_{\mathbf{a}}}l_{\mathbf{c}}\right) - \mathcal{L}_{Y^{f}_{\mathbf{a}}}\left(\mathcal{L}_{X_{\mathbf{b}}}l_{\mathbf{c}}\right)= \\
&= \mathcal{L}_{ X_{\mathbf{b}}}\left(G_{f}\left(Y^{f}_{\mathbf{a}},Y^{f}_{\mathbf{c}}\right)\right) - \mathcal{L}_{Y^{f}_{\mathbf{a}}}\left(l_{[\mathbf{b},\mathbf{c}]}\right)= \\
&=  G_{f}\left([ X_{\mathbf{b}},Y^{f}_{\mathbf{a}}],Y^{f}_{\mathbf{c}}\right) + G_{f}\left(Y^{f}_{\mathbf{a}}, [ X_{\mathbf{b}},Y^{f}_{\mathbf{c}}]\right)  - \mathcal{L}_{Y^{f}_{\mathbf{a}}}\left(l_{[\mathbf{b},\mathbf{c}]}\right)= \\
&=  \mathcal{L}_{[X_{\mathbf{b}},Y^{f}_{\mathbf{a}}]}l_{\mathbf{c}} + G_{f}\left(Y^{f}_{\mathbf{a}}, [ X_{\mathbf{b}},Y^{f}_{\mathbf{c}}]\right)  - G_{f}\left(Y^{f}_{\mathbf{a}},Y^{f}_{[\mathbf{b},\mathbf{c}]}\right)  \\
\end{split}
\ee
where we used the equation \eqref{eqn: fundamental vector fields of U(H) on expectation value functions}, and the  fact that $\mathcal{L}_{X_{\mathbf{a}}} G_{f}=0$ because the fundamental vector fields of the action $\alpha$ of $\Uh$ are Killing vector fields for all monotone quantum metric tensors.
%%%%%
From equation \eqref{eqn: constraints mixed brackets} we conclude that
\be\label{eqn: constraints mixed brackets 2}
G_{f}\left(Y^{f}_{\mathbf{a}}, [ X_{\mathbf{b}},Y^{f}_{\mathbf{c}}]-Y^{f}_{[\mathbf{b},\mathbf{c}]}\right)=0
\ee
for every $\mathbf{a},\mathbf{b},\mathbf{c}\in\bhsa$.
%%%%
Then, since the differential of the expectation value functions provide a basis for the differential forms on $\stsph$, equation \eqref{eqn: constraints mixed brackets 2} is equivalent to
\be\label{eqn: constraints mixed brackets 3}
[ X_{\mathbf{b}},Y^{f}_{\mathbf{c}}]=Y^{f}_{[\mathbf{b},\mathbf{c}]}.
\ee
%%%%%
Equation \eqref{eqn: constraints mixed brackets 3} fixes the Lie bracket between elements of $\uh\oplus\{\mathbf{0}\}$ and its complement, thus leaving us with the freedom to only define the bracket among elements that lies in the complement of $\uh\oplus\{\mathbf{0}\}$ inside the Lie algebra $\mathfrak{g}$ of $\gapp$.
%%%%%%%%%%%%%%%%%%%%%%%%%%%%%%%%%%%%%%

%%%%%%%%%%%%%%%%%%%%%%%%%%%%%%%%%%%%%%%%%%%%%%%%%%%%%%%%%%%%
We are currently investigating all the problems discussed in this section and we plan to address them in detail in the (hopefully not too distant) future.
%%%%%

\addcontentsline{toc}{section}{References}
%\footnotesize{
\bibliographystyle{plain}
\bibliography{scientific_bibliography}

\begin{thebibliography}{10}

\bibitem{A-M-R-2012}
R.~Abraham, J.~E. Marsden, and T.~Ratiu.
\newblock {\em {Manifolds, tensor analysis, and applications}}.
\newblock Springer-Verlag, New York, third edition, 2012.

\bibitem{A-G-M-M-1994}
D.~Alekseevsky, J.~Grabowski, G.~Marmo, and P.~Michor.
\newblock {Poisson structures on the cotangent bundle of a Lie group or a
  principle bundle and their reductions}.
\newblock {\em Journal of Mathematical Physics}, 35(9):4909--49027, 1994.
\newblock {\href{https://doi.org/10.1063/1.530822}{DOI: 10.1063/1.530822}}.

\bibitem{A-G-M-M-1998}
D.~Alekseevsky, J.~Grabowski, G.~Marmo, and P.~Michor.
\newblock {Poisson structures on double Lie groups}.
\newblock {\em Journal of Geometry and Physics}, 26(3-4):340--379, 1998.
\newblock {\href{https://doi.org/10.1016/S0393-0440(97)00063-6}{DOI:
  10.1016/S0393-0440(97)00063-6}}.

\bibitem{A-N-2000}
S.~I. Amari and H.~Nagaoka.
\newblock {\em {Methods of Information Geometry}}.
\newblock American Mathematical Society, Providence, RI, 2000.
\newblock {\href{https://doi.org/10.1090/mmono/191}{DOI: 10.1090/mmono/191}}.

\bibitem{A-L-2021}
A.~Andai and A.~Lovas.
\newblock {Quantum Aitchison geometry}.
\newblock {\em Infinite Dimensional Analysis, Quantum Probability and Related
  Topics}, 01(24):2150001, 2021.
\newblock {\href{https://doi.org/10.1142/S0219025721500016}{DOI:
  10.1142/S0219025721500016}}.

\bibitem{A-S-1999}
A.~Ashtekar and T.~A. Schilling.
\newblock Geometrical formulation of quantum mechanics.
\newblock In A.~Harvey, editor, {\em On Einstein's Path: Essays in Honor of
  Engelbert Schucking}, pages 23 -- 65. Springer-Verlag, New York, 1999.

\bibitem{A-J-L-S-2017}
N.~Ay, J.~Jost, H.~V. Le, and L.~Schwachh\"{o}fer.
\newblock {\em {Information Geometry}}.
\newblock Springer International Publishing, 2017.
\newblock {\href{https://doi.org/10.1007/978-3-319-56478-4}{DOI:
  10.1007/978-3-319-56478-4}}.

\bibitem{B-B-M-2016}
M.~Bauer, M.~Bruveris, and P.~W. Michor.
\newblock {Uniqueness of the Fisher–Rao metric on the space of smooth
  densities}.
\newblock {\em Bulletin of the London Mathematical Society}, 48(3):499--506,
  2016.
\newblock {\href{https://doi.org/10.1112/blms/bdw020}{DOI:
  10.1112/blms/bdw020}}.

\bibitem{B-Z-2006}
I.~Bengtsson and K.~\.Zyczkowski.
\newblock {\em {Geometry of Quantum States: An Introduction to Quantum
  Entanglement}}.
\newblock Cambridge University Press, New York, 2006.
\newblock {\href{https://doi.org/10.1017/cbo9780511535048}{DOI:
  10.1017/cbo9780511535048}}.

\bibitem{Bhatia-2007}
R.~Bhatia.
\newblock {\em {Positive Definite Matrices}}.
\newblock Princeton University Press, 2007.

\bibitem{Blackadar-2006}
B.~Blackadar.
\newblock {\em {Operator Algebras: Theory of $C^*$-algebras and von Neumann
  Algebras}}.
\newblock Springer-Verlag, Berlin, 2006.

\bibitem{B-R-1987-1}
O.~Bratteli and D.~W. Robinson.
\newblock {\em {Operator Algebras and Quantum Statistical Mechanics I}}.
\newblock Springer-Verlag, Berlin, second edition, 1987.
\newblock {\href{https://doi.org/10.1007/978-3-662-03444-6}{DOI:
  10.1007/978-3-662-03444-6}}.

\bibitem{Bures-1969}
D.~Bures.
\newblock {An Extension of Kakutani's Theorem on Infinite Product Measures to
  the Tensor Product of Semifinite W*-Algebras}.
\newblock {\em Transactions of the American Mathematical Society},
  135:199--212, 1969.
\newblock {\href{https://doi.org/10.2307/1995012}{DOI: 10.2307/1995012}}.

\bibitem{Campbell-1986}
L.~L. Campbell.
\newblock An extended cencov characterization of the information metric.
\newblock {\em Proceedings of the American Mathematical Society}, 98(1):135 --
  141, 1986.

\bibitem{Cencov-1982}
N.~N. Cencov.
\newblock {\em {Statistical Decision Rules and Optimal Inference}}.
\newblock American Mathematical Society, Providence, RI, 1982.
\newblock {\href{https://doi.org/10.1090/mmono/053}{DOI: 10.1090/mmono/053}}.

\bibitem{Choi-1975}
M.~Choi.
\newblock {Completely positive linear maps on complex matrices}.
\newblock {\em Linear Algebra and its Applications}, 10(3):285--290, 1975.
\newblock {\href{https://doi.org/10.1016/0024-3795(75)90075-0}{DOI:
  10.1016/0024-3795(75)90075-0}}.

\bibitem{Ciaglia-2020}
F.~M. Ciaglia.
\newblock {Quantum states, groups and monotone metric tensors}.
\newblock {\em European Physical Journal Plus}, 135:530--16pp, 2020.
\newblock {\href{https://doi.org/10.1140/epjp/s13360-020-00537-y}{DOI:
  10.1140/epjp/s13360-020-00537-y}}.

\bibitem{C-DC-DN-V-2022}
F.~M. Ciaglia, F.~Di~Cosmo, F.~Di~Nocera, and P.~Vitale.
\newblock {Monotone metric tensors in Quantum Information Geometry}.
\newblock 2022.
\newblock {\href{https://doi.org/10.48550/arXiv.2203.10857}{DOI:
  10.48550/arXiv.2203.10857}}.

\bibitem{C-DC-I-L-M-2017}
F.~M. Ciaglia, F.~Di~Cosmo, A.~Ibort, M.~Laudato, and G.~Marmo.
\newblock {Dynamical vector fields on the manifold of quantum states}.
\newblock {\em Open Systems \& Information dynamics}, 24(3):1740003--38, 2017.
\newblock {\href{https://doi.org/10.1142/S1230161217400030}{DOI:
  10.1142/S1230161217400030}}.

\bibitem{C-DC-L-M-M-V-V-2018}
F.~M. Ciaglia, F.~Di~Cosmo, M.~Laudato, G.~Marmo, G.~Mele, F.~Ventriglia, and
  P.~Vitale.
\newblock {A Pedagogical Intrinsic Approach to Relative Entropies as Potential
  Functions of Quantum Metrics: the q-z family}.
\newblock {\em Annals of Physics}, 395:238--274, 2018.
\newblock {\href{https://doi.org/10.1016/j.aop.2018.05.015}{DOI:
  10.1016/j.aop.2018.05.015}}.

\bibitem{C-DN-2021}
F.~M. Ciaglia and F.~Di~Nocera.
\newblock {Group Actions and Monotone Metric Tensors: The Qubit Case}.
\newblock In Frank Nielsen and Fr{\'e}d{\'e}ric Barbaresco, editors, {\em
  Geometric Science of Information 2021}, volume 12829 of {\em Lecture Notes in
  Computer Science}, pages 145--153. Springer International Publishing, 2021.
\newblock {\href{https://doi.org/10.1007/978-3-030-80209-7_17}{DOI:
  10.1007/978-3-030-80209-717}}.

\bibitem{C-I-J-M-2019}
F.~M. Ciaglia, A.~Ibort, J.~Jost, and G.~Marmo.
\newblock {Manifolds of classical probability distributions and quantum density
  operators in infinite dimensions}.
\newblock {\em Information Geometry}, 2(2):231--271, 2019.
\newblock {\href{https://doi.org/10.1007/s41884-019-00022-1}{DOI:
  10.1007/s41884-019-00022-1}}.

\bibitem{C-J-S-2020-02}
F.~M. Ciaglia, J.~Jost, and L.~Schwachh\"{o}fer.
\newblock {Differential geometric aspects of parametric estimation theory for
  states on finite-dimensional C*-algebras}.
\newblock {\em Entropy}, 22(11):1332, 2020.
\newblock {\href{https://doi.org/10.3390/e22111332}{DOI: 10.3390/e22111332}}.

\bibitem{C-J-S-2020}
F.~M. Ciaglia, J.~Jost, and L.~Schwachh\"{o}fer.
\newblock {From the Jordan product to Riemannian geometries on classical and
  quantum states}.
\newblock {\em Entropy}, 22(06):637--27, 2020.
\newblock {\href{https://doi.org/10.3390/e22060637}{DOI: 10.3390/e22060637}}.

\bibitem{DA-F-2021}
F.~D'Andrea and D.~Franco.
\newblock {On the pseudo-manifold of quantum states}.
\newblock {\em Differential Geometry and its Applications}, 78:101800, 2021.
\newblock {\href{https://doi.org/10.1016/j.difgeo.2021.101800}{DOI:
  10.1016/j.difgeo.2021.101800}}.

\bibitem{Dirac-1958}
P.~A.~M. Dirac.
\newblock {\em {Principles of Quantum Mechanics}}.
\newblock Oxford university Press, 1958.

\bibitem{Dittmann-1995}
J.~Dittmann.
\newblock {On the Riemannian metric on the space of density matrices}.
\newblock {\em Reports on Mathematical Physics}, 36(3):309--315, 1995.
\newblock {\href{https://doi.org/10.1016/0034-4877(96)83627-5}{DOI:
  10.1016/0034-4877(96)83627-5}}.

\bibitem{Fisher-1922}
R.~A. Fisher.
\newblock {On the mathematical foundations of theoretical statistics}.
\newblock {\em Philosophical Transactions of the Royal Society of London.
  Series A}, 222:309 -- 368, 1922.

\bibitem{Fujiwara-1999}
A.~Fujiwara.
\newblock {Geometry of Quantum Information Systems}.
\newblock In O.~E. Barndorff-Nielsen and E.~B.~V. Jensen, editors, {\em
  {Geometry in Present Day Science}}, pages 35--48, 1999.
\newblock {\href{https://doi.org/10.1142/3958}{DOI: 10.1142/3958}}.

\bibitem{Furuta-2008}
T.~Furuta.
\newblock {Concrete examples of operator monotone functions obtained by an
  elementary method without appealing to Löwner integral representation}.
\newblock {\em Linear Algebra and its Applications}, 429(5--6):972--980, 2008.
\newblock {\href{https://doi.org/10.1016/j.laa.2006.11.023}{DOI:
  10.1016/j.laa.2006.11.023}}.

\bibitem{G-I-2001}
P.~Gibilisco and T.~Isola.
\newblock {A characterization of Wigner-Yanase skew information among
  statistically monotone metrics}.
\newblock {\em {Infinite Dimensional Analysis, Quantum Probability and Related
  Topics}}, 4(4):553--557, 2001.
\newblock {\href{https://doi.org/10.1142/s0219025701000644}{DOI:
  10.1142/s0219025701000644}}.

\bibitem{G-I-2003}
P.~Gibilisco and T.~Isola.
\newblock {Wigner-Yanase information on quantum state space: the geometric
  approach}.
\newblock {\em Journal of Mathematical Physics}, 44(9):3752--3762, 2003.
\newblock {\href{https://doi.org/10.1063/1.1598279}{DOI: 10.1063/1.1598279}}.

\bibitem{G-K-M-2005}
J.~Grabowski, M.~Ku{\'s}, and G.~Marmo.
\newblock {Geometry of quantum systems: density states and entanglement}.
\newblock {\em Journal of Physics A: Mathematical and General},
  38(47):10217--10244, 2005.
\newblock {\href{https://doi.org/10.1088/0305-4470/38/47/011}{DOI:
  10.1088/0305-4470/38/47/011}}.

\bibitem{G-K-M-2006}
J.~Grabowski, M.~Ku{\'s}, and G.~Marmo.
\newblock {Symmetries, group actions, and entanglement}.
\newblock {\em Open Systems \& Information Dynamics}, 13(04):343 -- 362, 2006.

\bibitem{Hasegawa-1993}
H.~Hasegawa.
\newblock {$\alpha$-Divergence of the non-commutative information geometry}.
\newblock {\em Reports in Mathematical Physics}, 33(1/2):87--93, 1993.
\newblock {\href{https://doi.org/10.1016/0034-4877(93)90043-E}{DOI:
  10.1016/0034-4877(93)90043-E}}.

\bibitem{Hasegawa-1995}
H.~Hasegawa.
\newblock {\em {Non-Commutative Extension of the Information Geometry}}, pages
  327--337.
\newblock Springer US, 1995.
\newblock {\href{https://doi.org/10.1007/978-1-4899-1391-3_31}{DOI:
  10.1007/978-1-4899-1391-331}}.

\bibitem{Hasegawa-2003}
H.~Hasegawa.
\newblock {Dual geometry of the Wigner-Yanase-Dyson information content}.
\newblock {\em Infinite Dimensional Analysis, Quantum Probability and Related
  Topics}, 6(3):413--430, 2003.
\newblock {\href{https://doi.org/10.1142/S021902570300133X}{DOI:
  10.1142/S021902570300133X}}.

\bibitem{H-P-1997}
H.~Hasegawa and D.~Petz.
\newblock {\em {Non-Commutative Extension of Information Geometry II}}, pages
  109 -- 118.
\newblock Springer, New York, 1997.
\newblock {\href{https://doi.org/10.1007/978-1-4615-5923-8_12}{DOI:
  10.1007/978-1-4615-5923-812}}.

\bibitem{Helstrom-1967}
C.~W. Helstrom.
\newblock {Minimum mean-squared error of estimates in quantum statistics}.
\newblock {\em Physics Letters A}, 25(2):101--102, 1967.
\newblock {\href{https://doi.org/10.1016/0375-9601(67)90366-0}{DOI:
  10.1016/0375-9601(67)90366-0}}.

\bibitem{Helstrom-1968}
C.~W. Helstrom.
\newblock {The minimum variance of estimates in quantum signal detection}.
\newblock {\em IEEE Transactions on Information Theory}, 14(2):234--242, 1968.
\newblock {\href{https://doi.org/10.1109/TIT.1968.1054108}{DOI:
  10.1109/TIT.1968.1054108}}.

\bibitem{Holevo-2001}
A.~S. Holevo.
\newblock {\em {Statistical Structure of Quantum Theory}}.
\newblock Springer-Verlag, Berlin, 2001.
\newblock {\href{https://doi.org/10.1007/3-540-44998-1}{DOI:
  10.1007/3-540-44998-1}}.

\bibitem{Holevo-2011}
A.~S. Holevo.
\newblock {\em {Probabilistic and Statistical Aspects of Quantum Theory}}.
\newblock Edizioni della Normale, 2011.
\newblock {\href{https://doi.org/10.1007/978-88-7642-378-9}{DOI:
  10.1007/978-88-7642-378-9}}.

\bibitem{Jencova-2003-2}
A.~Jen\v{c}ov\'{a}.
\newblock {Flat connections and Wigner-Yanase-Dyson metrics}.
\newblock {\em Reports on Mathematical Physics}, 52(3):331--351, 2003.
\newblock {\href{https://doi.org/10.1016/S0034-4877(03)80033-2}{DOI:
  10.1016/S0034-4877(03)80033-2}}.

\bibitem{Kumagai-2011}
W.~Kumagai.
\newblock {A characterization of extended monotone metrics}.
\newblock {\em Linear Algebra and its Applications}, 434(1):224--231, 2011.
\newblock {\href{https://doi.org/10.1016/j.laa.2010.08.019}{DOI:
  10.1016/j.laa.2010.08.019}}.

\bibitem{Landsman-2017}
N.~P. Landsman.
\newblock {\em {Foundations of Quantum Theory. From Classical Concepts to
  Operator Algebras}}.
\newblock Springer International Publishing, Cham, Switzerland, 2017.
\newblock {\href{https://doi.org/10.1007/978-3-319-51777-3}{DOI:
  10.1007/978-3-319-51777-3}}.

\bibitem{Loewner-1934}
Karl L{\"o}wner.
\newblock {\"U}ber monotone matrixfunktionen.
\newblock {\em Mathematische Zeitschrift}, 38(1):177--216, 1934.

\bibitem{Mahalanobis-1936}
P.~C. Mahalanobis.
\newblock {On the generalized distance in Statistics}.
\newblock {\em Proceedings of the National Institute of Sciences of India},
  II(1):49--55, 1936.

\bibitem{M-M-V-V-2017}
V.~I. Man'ko, G.~Marmo, F.~Ventriglia, and P.~Vitale.
\newblock {Metric on the space of quantum states from relative entropy.
  Tomographic reconstruction}.
\newblock {\em Journal of Physics A: Mathematical and Theorerical},
  50(33):335302, 2017.
\newblock {\href{https://doi.org/10.1088/1751-8121/aa7d7d}{DOI:
  10.1088/1751-8121/aa7d7d}}.

\bibitem{Naudts-2018}
J.~Naudts.
\newblock {Quantum statistical manifolds}.
\newblock {\em Entropy}, 20(6), 2018.
\newblock {\href{https://doi.org/10.3390/e20060472}{DOI: 10.3390/e20060472}}.

\bibitem{Naudts-2022}
J.~Naudts.
\newblock {Exponential arcs in the manifold of vector states on a
  $\sigma$-finite von Neumann algebra}.
\newblock {\em Information Geometry}, 2022.
\newblock {\href{https://doi.org/10.1007/s41884-021-00064-4}{DOI:
  10.1007/s41884-021-00064-4}}.

\bibitem{N-V-W-1975}
J.~Naudts, A.~Verbeure, and R.~Weder.
\newblock {Linear Response Theory and the KMS Condition}.
\newblock {\em Communications in Mathematical Physics}, 44:87--99, 1975.
\newblock {\href{https://doi.org/10.1007/BF01609060}{DOI: 10.1007/BF01609060}}.

\bibitem{Naudts-2021}
{Naudts, Jan}.
\newblock {Parameter-free description of the manifold of non-degenerate density
  matrices}.
\newblock {\em European Physical Journal Plus}, 136(1), 2021.
\newblock {\href{https://doi.org/10.1140/epjp/s13360-020-01038-8}{DOI:
  10.1140/epjp/s13360-020-01038-8}}.

\bibitem{N-C-2011}
M.~A. Nielsen and I.~L. Chuang.
\newblock {\em {Quantum Computation and Quantum Information}}.
\newblock Cambridge University Press, New York, NY, 2011.

\bibitem{Petz-1994}
D.~Petz.
\newblock Geometry of canonical correlation on the state space of a quantum
  system.
\newblock {\em Journal of Mathematical Physics}, 35(2):780--795, 1994.
\newblock {\href{https://doi.org/10.1063/1.530611}{DOI: 10.1063/1.530611}}.

\bibitem{Petz-1996}
D.~Petz.
\newblock {Monotone metrics on matrix spaces}.
\newblock {\em Linear Algebra and its Applications}, 244:81--96, 1996.
\newblock {\href{https://doi.org/10.1016/0024-3795(94)00211-8}{DOI:
  10.1016/0024-3795(94)00211-8}}.

\bibitem{P-T-1993}
D.~Petz and G.~Toth.
\newblock {The Bogoliubov Inner Product in Quantum Statistics}.
\newblock {\em Letters in Mathematical Physics}, 27:205--216, 1993.
\newblock {\href{https://doi.org/10.1007/BF00739578}{DOI: 10.1007/BF00739578}}.

\bibitem{Rao-1945}
C.~R. Rao.
\newblock {Information and accuracy attainable in the estimation of statistical
  parameters}.
\newblock {\em Bulletin of the Calcutta Mathematical Society}, 37(3):81--91,
  1945.

\bibitem{Suzuki-2021}
J.~Suzuki.
\newblock {Non-monotone metric on the quantum parametric model}.
\newblock {\em European Physical Journal Plus}, 136(1), 2021.
\newblock {\href{https://doi.org/10.1140/epjp/s13360-021-01101-y}{DOI:
  10.1140/epjp/s13360-021-01101-y}}.

\bibitem{Suzuki-1997}
M.~Suzuki.
\newblock {Quantum Analysis - Non-Commutative Differential and Integral
  Calculi}.
\newblock {\em Communications in Mathematical Physics volume}, 183:339 -- 363,
  1997.

\bibitem{Takesaki-2002}
M.~Takesaki.
\newblock {\em {Theory of Operator Algebra I}}.
\newblock Springer-Verlag, Berlin, 2002.

\bibitem{Uhlmann-1992}
A.~Uhlmann.
\newblock {\em Groups and related topics}, chapter {The Metric of Bures and the
  Geometric Phase}, pages 267--274.
\newblock Springer, Dordrecht, 1992.
\newblock {\href{https://doi.org/10.1007/978-94-011-2801-8_23}{DOI:
  10.1007/978-94-011-2801-823}}.

\bibitem{von-Neumann-1955}
J.~von Neumann.
\newblock {\em {Mathematical Foundations of Quantum Mechanics}}.
\newblock Princeton University Press, Princeton, NJ, 1955.

\end{thebibliography}
%}
\end{document}